\documentclass[12pt]{article}

\usepackage{fullpage}
\usepackage[T1]{fontenc}
\usepackage{ae,aecompl}
\usepackage[dvips]{graphicx}
\usepackage{amssymb}
\usepackage{amsmath}
\usepackage{subfigure}
\usepackage[round,authoryear]{natbib}
\usepackage{color}
\usepackage{array}
\usepackage{rotating}
\usepackage{colortbl}
\usepackage{tikz}
\usetikzlibrary{patterns}

\definecolor{webgreen}{rgb}{0,0.4,0}
\definecolor{webbrown}{rgb}{0.6,0,0}
\definecolor{purple}{rgb}{0.5,0,0.25}
\definecolor{darkblue}{rgb}{0,0,0.7}
\definecolor{darkred}{rgb}{0.7,0,0}
\definecolor{darkgreen}{rgb}{0,0.7,0}
\usepackage{hyperref}
\hypersetup{colorlinks,citecolor=darkred,filecolor=black,linkcolor=darkblue,urlcolor=webgreen}
\newcommand{\ignore}[1]{}
\newtheorem{lemma}{{\sc Lemma}}

\newtheorem{prop}{{\sc Proposition}}
\newtheorem{cor}{{\sc Corollary}}
\newtheorem{theorem}{{\sc Theorem}}
\newtheorem{defn}{{\sc Definition}}

\newtheorem{claim}{{\sc Claim}}
\newtheorem{example}{{\sc Example}}
\newtheorem{fact}{{\sc Fact}}

\newenvironment{proof}{\noindent {\bf \sl Proof\/}:\enspace}
{\hfill $\blacksquare{}$ \vspace{12pt}}
\usepackage{sectsty}
\sectionfont{\fontsize{14}{14}\centering\normalfont\scshape}
\subsectionfont{\fontsize{14}{14}\centering\normalfont}
\subsubsectionfont{\fontsize{12}{12}\centering\normalfont}

\newenvironment{exam}{\noindent \enspace}
{\hfill $\lozenge$ \vspace{12pt}}

\usepackage{cleveref}

\title{{\Large {\bf Pareto efficient combinatorial auctions:} \\
{\bf dichotomous preferences without quasilinearity}}\thanks{We thank two anonymous referees and the associate editor for their detailed comments. We are grateful to
Krishnendu Ghosh Dastidar, Tomoya Kazumura, Dilip Mookherjee, Kolagani Paramahamsa, Arunava Sen, Shigehiro Serizawa and seminar participants
at ISI Delhi for useful comments.}}

\author{{\small
Komal Malik and Debasis Mishra}\thanks{Komal Malik, Indian Statistical Institute, New Delhi, Email: \texttt{komal.malik1815r@isid.ac.in};
Debasis Mishra, Indian Statistical Institute, New Delhi, Email: \texttt{dmishra@isid.ac.in}}}

\date{\small{\today}}

\begin{document}
\maketitle

\begin{abstract}

We consider a combinatorial auction model where preferences of agents over bundles of objects and payments need not be quasilinear.
However, we restrict the preferences of agents to be dichotomous.
An agent with dichotomous preference partitions the set of bundles of objects as {\em acceptable} and {\em unacceptable}, and at the same
payment level, she is indifferent between bundles in each class but strictly prefers acceptable to unacceptable bundles.
We show that there is no Pareto efficient,
dominant strategy incentive compatible (DSIC), individually rational (IR) mechanism satisfying no subsidy if the domain of preferences
includes {\em all} dichotomous preferences.
However, a generalization of the VCG mechanism is Pareto efficient, DSIC, IR and satisfies no subsidy
if the domain of preferences contains only {\em positive income effect} dichotomous preferences.
We show tightness of this result: adding any non-dichotomous
preference (satisfying some natural properties) to the domain of quasilinear dichotomous preferences brings back the impossibility result. \\



\end{abstract}

\section{Introduction}

The Vickrey-Clarke-Groves (VCG) mechanism~\citep{Vickrey61,Clarke71,Groves73} occupies a central role in mechanism design theory (specially, with
private values).
It satisfies two fundamental desiderata: it is dominant strategy incentive compatible (DSIC) and Pareto efficient.
We study a model of combinatorial auctions, where multiple objects are sold to agents simultaneously, who may
buy any bundle of objects. For such combinatorial auction models, the VCG mechanism and its indirect implementations (like ascending
price auctions)
have been popular.
The VCG mechanism is also individually rational (IR) and satisfies no subsidy (i.e., does not subsidize any agent) in these models.

Unfortunately, these desirable properties of the
VCG mechanism critically rely on the fact that agents have quasilinear preferences. While analytically convenient and a good approximation of actual
preferences when payments involved are low, quasilinearity is a debatable assumption in practice. For instance, consider an agent
participating in a combinatorial auction for spectrum licenses, where agents often borrow from various investors at non-negligible interest
rates. Such borrowing naturally leads to a preference which is not quasilinear. Further, income effects are ubiquitous in settings with non-negligible payments. For instance, a bidder in a spectrum auction often needs to invest in telecom infrastrastructure to
realize the full value of spectrum. Higher payment in the auction will lead to lower investments in infrastructure, and hence,
a lower value for the spectrum.

This has initiated a small literature in mechanism design theory (discussed later in this section and again in Section \ref{sec:lit}), where
the quasilinearity assumption is relaxed to allow any {\em classical} preference of the agent over consumption bundles: (bundle of objects, payment) pairs.\footnote{
Classical preferences assume mild continuity and monotonicity (in money and bundles of objects) properties of preferences.}
The main research question addressed in this literature is
the following:
\begin{quote}
{\sl In combinatorial auction models, if agents have classical
preferences, is it possible to construct a ``desirable" mechanism: a mechanism which inherits
the DSIC, Pareto efficiency, IR, and no subsidy properties of the VCG mechanism?}
\end{quote}

\subsection{Dichotomous preferences}

This paper contributes to this literature, focusing on the particular case in which agents' preferences
belong to a class of preferences, which we call {\em dichotomous}.
If an agent has a dichotomous preference, she partitions the set of bundles of objects
into {\em acceptable} and {\em unacceptable}. If the payments for all the bundles of objects are the same, then an agent is indifferent between her acceptable bundles of objects; she
is also indifferent between unacceptable bundles of objects; but she prefers every acceptable bundle to every unacceptable bundle.

Such preferences, though restrictive, are found in many settings of interest. For instance, consider the recent ``incentive auction"
done by the US Government~\citep{Leyton17}. It involved a ``reverse auction" phase where the broadcast licenses from existing broadcasters
were bought; a ``forward auction" phase where buyers bought broadcast licenses; and a clearing phase. The auction resulted
in billions of dollars in revenue for US treasury \citep{Leyton17}.
The theoretical analysis
of the reverse auction phase was done by \citet{Milgrom19}, where they assume quasilinear preferences with ``single-minded"
bidders, a specific kind of dichotomous preference
where the bidder has a {\em unique} acceptable bundle (a broadcast band in this case).
In these auctions, a broadcaster had some feasible frequency bands in which it can operate.
Any of those feasible frequency bands were ``acceptable'' and it was indifferent between them (since
any of these frequencies allowed the broadcaster to realize its full value of broadcast).
This resulted in dichotomous preferences of agents.\footnote{Quoting \citet{Milgrom17b},
``Milgrom and Segal (2015) (hereafter MS) offer a theoretical analysis which assumes that all
bidders are single-station owners who know their station values and are ``single-minded", that is, willing
to bid only for a single option. This assumption is reasonable for commercial UHF broadcasters that view VHF
bands as ill-suited for their operations and for non-profit broadcasters that are willing to move for
compensation to a particular VHF band but that view going off-air as incompatible with their mission."}
\citet{Milgrom19} argue that the VCG mechanism
is computationally challenging in this setting and propose a simpler mechanism.

The assumption of dichotomous preferences seems natural in settings where a bidder is acquiring some resources, and finds any bundle
acceptable if it satisfies some requirements.
For instance, consider the following examples.

\begin{itemize}

\item Consider a scheduling problem, where a certain set of jobs (say, flights at the take-off slots
of an airport) need to be scheduled on
a server. There are certain intervals where each job is available and can be processed and other intervals
are not acceptable. For instance, a supplier bidding to supply to a firm's production schedule can do so
only on some fixed interval of dates. So, certain dates are acceptable to it and others are not acceptable.
A traveller is buying tickets between a pair of cities but find certain dates acceptable for travel and realize
value only on those dates.

\item Consider a seller who is selling land to different buyers. The lands differ in size but are homogeneous otherwise.
Each buyer only demands a land of a fixed size. For instance, suppose the Government is allocating land to firms to set up
factories in a region, and each firm needs a land of a fixed size to set up its factory. This means all the bundles
of land exceeding the size requirement are acceptable to a firm.

\item Consider firms (data providers) buying paths on (data) networks~\citep{Babaioff09} - a firm is interested in sending
data from node $x$ to node $y$ on a directed graph whose edges are up for sale, and as long as a bundle of edges
contain a path from $x$ to $y$, it is acceptable to the
firm.

\end{itemize}

In all the examples above, if the payment involved are high, we can expect income effects, which will
mean that agents do not have quasilinear preferences.
One may also consider the dichotomous preference restriction as a behavioural assumption, where the agent does not consider computing values
for each of the exponential number of bundles but classifies the bundles as acceptable and unacceptable.
Hence, they are easy to elicit even in combinatorial auction setting.
Even with quasilinear preferences, the dichotomous restriction poses interesting combinatorial challenges for computing the VCG outcome.
This has led to a large literature in computer science for looking at {\em approximately desirable} VCG-style mechanisms~\citep{Babaioff05,Babaioff09, Lehmann02, Ledyard07, Milgrom14}.
Also related is the literature in matching and social choice theory (models without payments), where dichotomous preferences have been widely studied~\citep{Bogo04,Bogo05,Bade15}.

\subsection{Summary and intuition of results}

We show that
if the domain of preferences contains {\em all} dichotomous classical preferences, there is no desirable mechanism.
However,
a natural generalization of the VCG mechanism to classical preferences, which we call the {\em generalized VCG} (GVCG) mechanism,
is desirable if the domain contains {\em only positive income effect} dichotomous preferences. In other words, when {\em normal} goods
are sold, the GVCG mechanism is desirable. Further,
the GVCG mechanism is the {\em unique} desirable mechanism in any domain of positive income effect dichotomous preferences if it contains
the quasilinear dichotomous preferences. The GVCG mechanism allocates the goods in a way such that the collective {\em willingness to pay} of all the bidders
is maximized. Classical preferences imply that willingness to pay for a bundle of objects depends on the payment level.
Thus, it is not clear what the counterpart of ``valuation" of a bundle of objects is in this setting. Our generalized VCG is defined
by treating the willingness to pay at {\em zero} payment as the ``valuation" of a bundle and then defining the VCG outcome
with respect to these valuations, i.e., the allocation maximizes the sum of agents' valuations and each agent
pays her externality.

The intuition for these results is the following. The GVCG mechanism allocates the goods in a way
that maximizes the collective willingness to pay of all the bidders. In fact, with enough richness
in the domain, every desirable mechanism must allocate objects like the GVCG mechanism at certain
profiles. Individual rationality implies
that winning bidders pay an amount less than their willingness to pay. So, winning makes a
winning bidder wealthier. With dichotomous preferences, the payments in the GVCG mechanism can be quite low.
If bidders have negative income effect, then their {\em willingness to sell} (i.e., the compensating amount
needed to make a winning bidder lose her bundle of objects) is lower than their willingness to pay.
This creates ex-post trading opportunities and the GVCG mechanism is no longer efficient.
On the other hand, with positive income effect, the willingness to sell of winning bidders
is higher than their willingness to pay and the GVCG mechanism is efficient.

Our positive result is tight: we get back impossibility in any domain containing quasilinear
dichotomous preferences and at least one more positive income effect non-dichotomous preference (satisfying some extra reasonable conditions).
Such an additional preference may be a unit-demand preference, where the agent is interested
in at most one object~\citep{Demange85}.
To get an intuition for this result, suppose we consider a domain which contains all quasilinear dichotomous preferences and
one unit-demand positive income effect preference. We know that the GVCG mechanism may not be strategy-proof
in the domain of unit-demand preferences if agents have income effects~\citep{Morimoto15}.
But, we know that in the quasilinear domain with dichotomous preferences, the GVCG mechanism
is the unique desirable mechanism. With two agents having positive income effect unit-demand preference and others having
quasilinear dichotomous preference, we show that the outcome in a desirable mechanism, if it existed, would still have to be the
outcome of the GVCG mechanism. As a result, the agents with positive income effect unit-demand preferences could manipulate
at such preference profiles. This negative result not only establishes the tightness of our
positive result, but also helps to illuminate the bigger picture of possibility and impossibility domains
without quasilinearity.

We briefly connect our results to some relevant results from the literature. A detailed literature survey
is given in Section \ref{sec:lit}.
\citet{Saitoh08} was the first paper to define the generalized VCG mechanism for the
single object auction model. They show
that the generalized VCG mechanism is desirable in their model even if preferences have
{\em negative income effect}. This is in contrast to our model, where we get impossibility with negative income effect
preferences but the generalized VCG mechanism is desirable with positive income effect.

When we go from single object to multiple object combinatorial auctions, the generalized VCG may
fail to be DSIC without quasilinear preferences. For instance, \citet{Demange85} consider a combinatorial auction model where
multiple heterogenous objects are sold but each agent demands at most one object.
In this model, the generalized VCG is no longer DSIC. However, \citet{Demange85}
propose a different mechanism (based on the idea of market-clearing prices), which is desirable.

When agents can demand more than one object in a combinatorial auction model with multiple
heterogeneous objects, \citet{Kazumura16} show that a desirable mechanism
may not exist - this result requires certain richness of the domain of preferences which is violated by our dichotomous
preference model. Similarly, \citet{Baisa17} shows that in the homogeneous objects sale case, if agents demand multiple
units, then a desirable mechanism may not exist -- he requires slightly different axioms than our desirability axioms.

These results point to a conjecture that when agents demand multiple objects in a combinatorial auction model,
a desirable mechanism may not exist. Since ours is a combinatorial auction model where agents can consume multiple objects,
an impossibility result might not seem surprising. However, dichotomous preferences are somewhat
close to the single object model preference. So, it is not clear which intuition dominates.
Our impossibility result with dichotomous preferences
complement the earlier impossibility results, showing that the multi-demand intuition goes through if
we include all possible dichotomous preferences. However, what is surprising is that we recover the
desirability of the generalized VCG mechanism with positive
income effect dichotomous preferences. This shows that {\em not all} multi-demand combinatorial auction models without quasilinearity
are impossibility domains.

\section{Preliminaries}
\label{sec:prelim}

Let $N=\{1,\ldots,n\}$ be the set of agents and $M$ be a set of $m$ objects.
Let $\mathcal{B}$ be the set of all subsets of $M$. We will refer to elements
in $\mathcal{B}$ as {\bf bundles} (of objects).
A seller (or a planner) is selling/allocating bundles from $\mathcal{B}$
to agents in $N$ using payments. We introduce the notion of classical preferences and type spaces corresponding
to them below.

\subsection{Classical Preferences}

Each agent has preference over possible {\em outcomes}, which are pairs of the
form $(A,t)$, where $A \in \mathcal{B}$ is a bundle and $t \in \mathbb{R}$ is the amount paid by the agent.
Let $\mathcal{Z}=\mathcal{B} \times \mathbb{R}$ denote
the set of all outcomes. A preference $R_i$ of agent $i$ over $\mathcal{Z}$
is a complete transitive preference relation with strict part denoted by $P_i$ and indifference
part denoted by $I_i$. This formulation of preference is very general and can
capture wealth effects. For instance, varying levels of transfers will correspond
to varying levels of wealth and this can be captured by our preference
over $\mathcal{Z}$.

We restrict attention to the following class of preferences.
\begin{defn}
Preference $R_i$ of agent $i$ over $\mathcal{Z}$ is {\bf classical} if it satisfies

\begin{enumerate}

\item {\bf Monotonicity}. for each $A,A' \in \mathcal{B}$ with $A' \subseteq A$ and for each $t,t' \in \mathbb{R}$ with $t' > t$,
the following hold:
(i) $(A,t)~P_i~(A,t')$ and (ii) $(A,t)~R_i~(A',t)$.

\item {\bf Continuity}. for each $Z \in \mathcal{Z}$, the upper contour set $\{Z' \in \mathcal{Z}: Z'~R_i~Z\}$ and the lower contour set $\{Z' \in \mathcal{Z}: Z~R_i~Z'\}$ are closed.

\item {\bf Finiteness}. for each $t \in \mathbb{R}$ and for each $A,A' \in \mathcal{B}$, there exist $ t' , t'' \in \mathbb{R}$ such that $(A',t')~R_i~(A,t)$ and $ (A,t)~R_i~(A',t'')$.
\end{enumerate}
\end{defn}
Restricting attention to such classical preferences is standard in mechanism design literature
without quasilinearity~\citep{Demange85,Baisa17,Morimoto15}.
The monotonicity conditions mentioned above are quite natural. The continuity and finiteness are technical conditions
needed to ensure nice structure of the indifference vectors. A quasilinear
preference is always classical, where {\em indifference vectors} are ``parallel".
Notice that the monotonicity condition requires a free-disposal property: at a fixed payment level, every bundle is weakly preferred to every other bundle which is a subset of it.
All our results continue to hold even if we relax this free-disposal property to require that at a fixed payment level, every bundle be weakly preferred to the empty bundle only.

Given a classical preference $R_i$, the {\bf willingness to pay (WP)} of agent $i$ at
$t$ for bundle $A$ is defined as the unique solution $x$ to the following equation:
$$(A,t+x)~I_i~(\emptyset,t).$$
We denote this solution as $WP(A,t;R_i)$. The following fact is immediate from monotonicity, continuity, and finiteness.
\begin{fact}
For every classical preference $R_i$, for every $A \in \mathcal{B}$ and for every $t \in \mathbb{R}$,
$WP(A,t;R_i)$ is a unique non-negative real number.
\end{fact}
For quasilinear preference, $WP(A,t;R_i)$ is independent of $t$ and represents the valuation for bundle $A$.

Another way to represent a classical preference is by a collection of indifference vectors. Fix a classical preference $R_i$.
Then, by definition, for every $t \in \mathbb{R}$ and for every $A \in \mathcal{B}$,
agent $i$ with classical preference $R_i$ will be indifferent between the following outcomes:
$$(\emptyset,t)~I_i~(A,t+WP(A,t;R_i)).$$

Figure \ref{fig:class} shows a representation of classical preference for three objects $\{a,b,c\}$.
The horizontal lines correspond to payment levels for each of the bundles. Hence, these lines are
the set of all outcomes $Z$ -- the space between these eight lines have no meaning and are
kept only for ease of illustration. As we go to the right along any of these lines, the outcomes become worse
since the payment (payment made by the agent) increases. Figure \ref{fig:class} shows eight points, each
corresponding to a unique bundle and a payment level for that bundle. These points are joined
to show that the agent is indifferent between these outcomes for a classical preference. Classical preference
implies that all the points to the left of this indifference vector are better than these outcomes and all the points
to the right of this indifference vector are worse than these outcomes. Indeed, every classical preference can be represented
by a collection of an infinite number of such indifference vectors. \\
\begin{figure}[!hbt]
\centering
\includegraphics[width=4.5in]{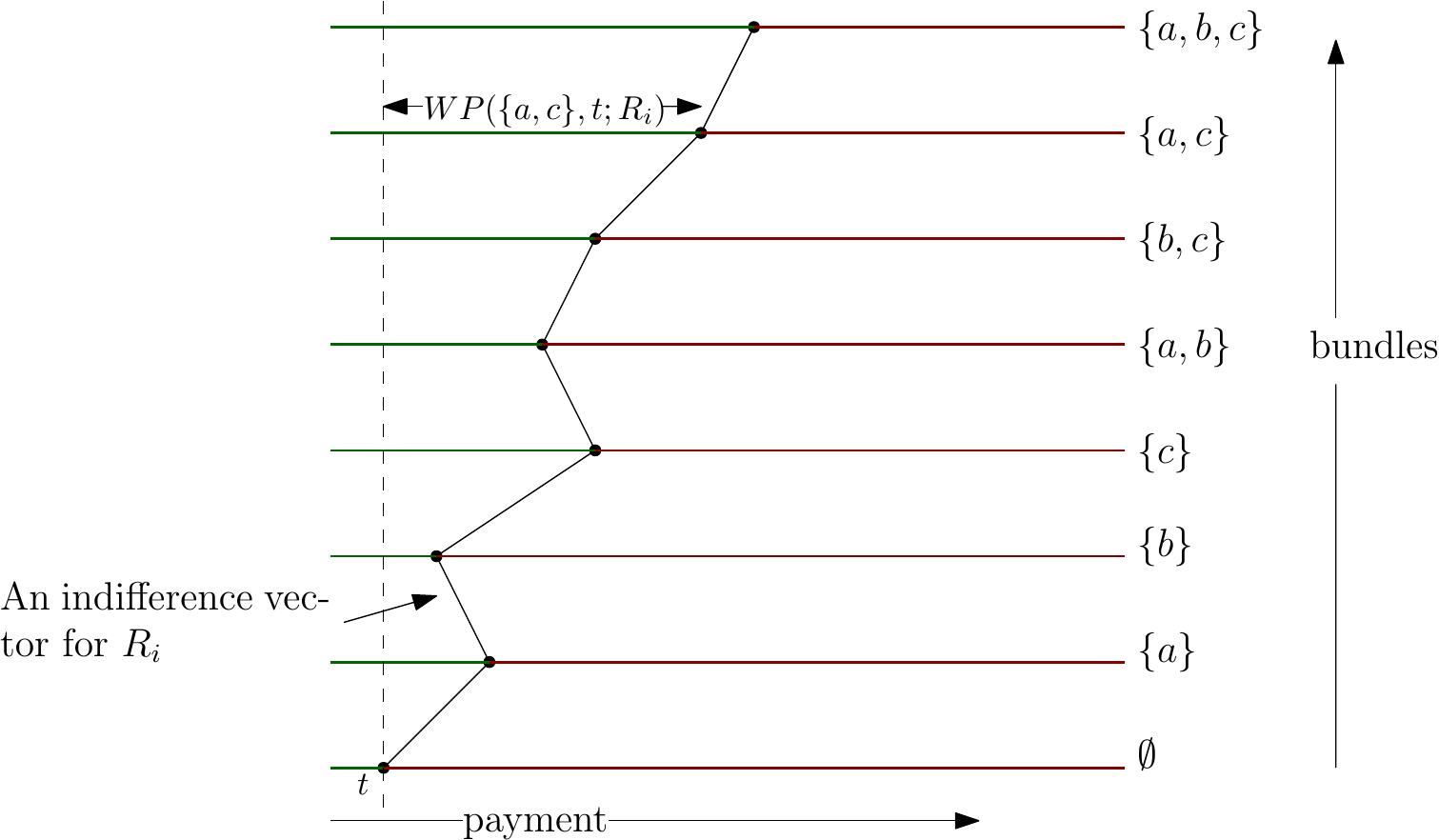}
\caption{Representation of classical preferences}
\label{fig:class}
\end{figure}

\subsection{Domains and mechanisms}

A {\em bundle allocation} is an ordered sequence of objects $(A_1,\ldots,A_n)$, where
$A_i$ denotes the bundle allocated to agent $i$, such that for each
$A_i,A_j \in \mathcal{B}$, we have $A_i \cap  A_j= \emptyset$ - note that $A_i$ can be equal to $\emptyset$ for any $i$ in an object
allocation.  Let $\mathcal{X}$ denote the set of all bundle allocations.

An outcome profile $((A_1,t_1),\ldots,(A_n,t_n))$ is a collection of $n$ outcomes
such that $(A_1,\ldots,A_n)$ is the bundle allocation and $t_i$ denotes the payment made by agent $i$.
An outcome profile $((A_1,t_1),\ldots,(A_n,t_n))$ is {\bf Pareto efficient} at $R \equiv (R_1,\ldots, R_n)$, if
there does not exist another outcome profile $((A'_1,t'_1),\ldots,(A'_n,t'_n))$ such that
\begin{enumerate}
\item for each $ i \in N, (A_i',t_i')~R_i~(A_i,t_i)$,
\item $\sum_{i \in N} t_i' \geq \sum_{i \in N}t_i$,
\end{enumerate}
with one of the inequalities strictly satisfied. The first relation says that each agent $i$ prefers
$(A'_i,t'_i)$ to $(A_i,t_i)$. The second relation requires that the seller is not spending money to make
everyone better off. Without the second relation, we can always improve any outcome profile by subsidizing the agents.\footnote{Our efficiency definition says that
the agents and the designer cannot improve using an outcome profile, which may involve negative
payments. Later, we
impose no-subsidy as an axiom for our mechanism.
The way to think about this is that Pareto efficient improvements are outside the
mechanism and may involve one agent or the designer ``buying" a bundle of objects from another
agent by compensating (negative payment) her.}

A {\bf domain or type space} is any subset of classical preferences. A typical domain of preferences will be denoted by $\mathcal{T}$.
A {\bf mechanism} is a pair $(f,\mathbf{p})$, where $f:\mathcal{T}^n \rightarrow \mathcal{X}$ and
$\mathbf{p} \equiv (p_1,\ldots,p_n)$ is a collection of payment rules with each $p_i:\mathcal{T}^n \rightarrow \mathbb{R}$.
Here, $f$ is the bundle allocation rule
and $p_i$ is the payment rule of agent $i$.
We denote the bundle allocated to agent $i$ at type profile $R$ by $f_i(R) \in \mathcal{B}$ in the bundle allocation rule $f$.

We require the following properties from a mechanism, which we term desirable.
\begin{defn}[Desirable mechanisms]
A mechanism $(f,\mathbf{p})$ is {\bf desirable} if

\begin{enumerate}
\item it is {\bf dominant strategy incentive compatible (DSIC)}: for all $i \in N$,
for all $R_{-i} \in \mathcal{T}^{n-1}$, and for all $R_i,R'_i \in \mathcal{T}$, we have
$$\Big(f_i(R),p_i(R)\Big)~R_i~\Big(f_i(R_i',R_{-i}),p_i(R'_i,R_{-i})\Big).$$

\item it is {\bf Pareto efficient}: $\Big((f_1(R),p_1(R)),\ldots,(f_n(R),p_n(R))\Big)$ is
Pareto efficient at $R$, for all $R \in \mathcal{T}^n$.

\item it is {\bf individually rational (IR)}: for all $R \in \mathcal{T}^n$ and for all $i \in N$,
$$\Big(f_i(R),p_i(R)\Big)~R_i~(\emptyset,0).$$

\item satisfies {\bf no subsidy}: for all $R \in \mathcal{T}^n$ and for all $ i \in N$,
$$p_i(R) \geq 0.$$

\end{enumerate}
\end{defn}

We will explore domains where a desirable mechanism exists. DSIC, Pareto efficiency,
and IR are standard constraints in mechanism design. No subsidy is
debatable. Our motivation for considering it as desirable stems from the fact that
most auction formats in practice and the VCG mechanism satisfy it. It also
discourages {\em fake} buyers from participating in the mechanism.

\subsection{A motivating example}
\label{sec:mot}

In this section, we provide an example to give some intuition for one of our main results.
\begin{example}
\label{ex:ex1}
\end{example}
\begin{exam}
Consider a setting with three agents $N=\{1,2,3\}$,
and two objects $M=\{a,b\}$. We are interested in a preference profile
where agents 2 and 3 have identical preference: $R_2=R_3=R_0$.
In particular, all non-empty bundles have the same willingness to pay according to $R_0$ and satisfy
$$WP(\{a,b\},t;R_0) = WP(\{a\},t;R_0) = WP(\{b\},t;R_0)=2+3t,$$
for $t > -\frac{1}{2}$. We are silent about the willingness to pay below $-\frac{1}{2}$, but it can be taken to
be $0.5$. We will only consider payments $t > -\frac{1}{2}$ for this example.
At preference $R_0$, we have
$$(\{a,b\},2+4t)~I_0~(\{b\},2+4t)~I_0~(\{a\},2+4t)~I_0~(\emptyset,t),$$
for all $t > -\frac{1}{2}$.
Hence, as $t$ increases, bundle $\{a\}$ (or $\{b\}$ or $\{a,b\}$) will require more payment to be indifferent to $(\emptyset,t)$.
We term this {\em negative income effect}.

\begin{table}[!hbt]
  \centering
  \begin{tabular}{|c||c | c | c|}
    \hline
     & $\{a\}$ & $\{b\}$ & $\{a,b\}$ \\
     \hline
     \hline
    $WP(\cdot,0;R_1)$ & $0$ & $0$ & $3.9$ \\
    \hline
      $WP(\cdot,0;R_2=R_0)$ & $2$ & $2$ & $2$ \\
      $WP(\cdot,0;R_3=R_0)$ & $2$ & $2$ & $2$ \\
    \hline
    \hline
  \end{tabular}
\caption{A profiles of preferences with $M=\{a,b\}$, $N=\{1,2,3\}$.}
\label{tab:negative}
\end{table}

Agent $1$ has {\em quasilinear} preference with
a value of $3.9$ for bundle $\{a,b\}$; value zero (or, arbitrarily close to zero) for bundle $\{a\}$ and bundle $\{b\}$, and
value of bundle $\emptyset$ is normalized to zero. We denote this preference as $R_1$.
The willingness to pay at zero payment for these
preferences are shown in Table \ref{tab:negative}.

Suppose $(f,\mathbf{p})$ is a desirable mechanism defined on a (rich enough) type space $\mathcal{T}$ containing the preference profile $R \equiv (R_1,R_2=R_0,R_3=R_0)$.
Notice that the value of $\{a,b\}$ for agent $1$ is 3.9 but $WP(\{a\},0;R_2)+WP(\{b\},0;R_3)=4$. Hence, a consequence
of Pareto efficiency, individual rationality, and no subsidy is that $f_1(R)=\emptyset$.\footnote{This follows from the following reasoning.
Individual rationality and no subsidy
imply that agents who are not allocated any object pay zero. Hence, any
outcome where agent $1$ is given both the objects
can be Pareto improved.} Then, without loss of generality, agent $2$ gets bundle $\{a\}$ and agent $3$ gets bundle $\{b\}$ due to Pareto efficiency.

Next, we can pin down the payments of agents at $R$. Since agent $1$ gets $\emptyset$, her payment must be zero by IR and no subsidy. Now, pretend as if agents $2$ and $3$ have quasilinear preference with valuations equal to their willingness to pay
at zero payment (see Table \ref{tab:negative}). Then, the VCG mechanism would charge them their externalities, which is equal to $1.9$ for both the
agents. If the type space $\mathcal{T}$ is sufficiently rich (in a sense, we make precise later),
DSIC will still require that $p_2(R)=p_3(R)=1.9$ (a precise argument is given in the proof of Theorem \ref{theo:impos}).

The negative income effect of $R_0$ makes the Pareto improvement possible in this example. The maximum payment
we can extract from agent $1$ is $3.9$. Hence, to collect more payment than the VCG outcome,
we can pay a maximum of $0.1 (=3.9-3.8)$ to agents 2 and 3.
If the preference $R_0$ was quasilinear, agents 2 and 3 would have required a compensation of $0.1$ each
to be indifferent between not getting any objects and the VCG outcome.
Due to negative income effect, agents 2 and 3 can be made to improve from their VCG outcome by paying
them much lower amounts. This in turn enables us to Pareto dominate the VCG outcome.

To be precise, the following outcome vector Pareto dominates the outcome of the mechanism at $R$:
$$z_1:= (\{a,b\},3.9),~~z_2:=(\emptyset,-0.025),~~z_3:=(\emptyset,-0.025).$$
To see why, note that (a) sum of payments in $z$ is $3.85 > p_2(R)+p_3(R)=3.8$; (b) agent $1$
is indifferent between $z_1$ and $(\emptyset,0)$; (c) agents 2 and 3 are also indifferent between
their outcomes in the mechanism and $z$ since $(\emptyset,-0.025)~I_0~(\{a\},1.9)$ (because
$WP(\{a\},t;R_0)=2+3t$ for all $t > -0.5$).

It is important to note that $R_1$ having high value on $\{a,b\}$ and (almost) zero value on all other bundles
played a crucial role in determining payments of agents, and hence, in the impossibility.
Indeed, if agent $1$ also had equal willingness to pay on some smaller bundle, then the example will not work.\footnote{
If the willingness to pay of agent $1$ is $3.9$ on $\{a\}$ or $\{b\}$, then her preference will
satisfy the {\em unit demand} property (for a formal definition, see Section \ref{sec:robust}).
Preference $R_0$ also satisfies the unit demand property.
It is known that if agents have unit demand preferences,
a desirable mechanism exists, even if such preferences have negative income effect~\citep{Demange85}.
}
This motivates the class of preferences we study in the next section.
\end{exam}

\subsection{Dichotomous preferences}

We turn our focus on a subset of classical preferences
which we call dichotomous. The dichotomous preferences can be described by: (a) a collection of
bundles, which we call the {\em acceptable} bundles, and (b) a willingness to pay function, which
only depends on the payment level. Formally, it is defined as follows.

\begin{defn}
A classical preference $R_i$ of agent $i$ is {\bf dichotomous} if there exists a non-empty set
of bundles $\emptyset \ne \mathcal{S}_i \subseteq (\mathcal{B} \setminus \{\emptyset\})$ and
a willingness to pay (WP) map $w_i:\mathbb{R} \rightarrow \mathbb{R}_{++}$ such that for every $t \in \mathbb{R}$,
\begin{displaymath}
WP(A,t;R_i) = \left\{ \begin{array}{l l}
w_i(t) & \forall~A \in \mathcal{S}_i \\
0 & \forall~A \in \mathcal{B} \setminus \mathcal{S}_i.
\end{array} \right.
\end{displaymath}
In this case, we refer to $\mathcal{S}_i$ as the collection of {\bf acceptable} bundles.
\end{defn}
The interpretation of the dichotomous preference is that, given same price (payment) for all the bundles,
the agent is indifferent between the bundles
in $\mathcal{S}_i$. Similarly, she is indifferent between the bundles in $\mathcal{B} \setminus \mathcal{S}_i$,
but it strictly prefers a bundle in $\mathcal{S}_i$ to a bundle outside it.
Hence, a dichotomous preference can be succinctly represented by a pair $(w_i,\mathcal{S}_i)$, where
$w_i: \mathbb{R} \rightarrow \mathbb{R}_{++}$ is a WP map and $\emptyset \ne \mathcal{S}_i \subseteq (\mathcal{B} \setminus \{\emptyset\})$ is the set of acceptable
bundles.

By our monotonicity requirement (free-disposal) of classical preference, for
every $S, T \in \mathcal{B}$, we have
$$\Big[ S \subseteq T, S \in \mathcal{S}_i \Big] \Rightarrow \Big[ T \in \mathcal{S}_i \Big].$$
Hence, a dichotomous preference can be described by $w_i$ and a {\em minimal}
set of bundles $\mathcal{S}^{min}_i$ such that
$$\mathcal{S}_i := \{T \in \mathcal{B}: S \subseteq T~\textrm{for some}~S \in \mathcal{S}^{min}_i\}.$$

Figure \ref{fig:dich}
shows two indifference vectors of a dichotomous preference. The figure shows that
the bundles $\{a\}, \{a,c\}, \{a,b\}$ and $\{a,b,c\}$ are acceptable but others are not. \\

\begin{figure}[!hbt]
\centering
\includegraphics[width=4.5in]{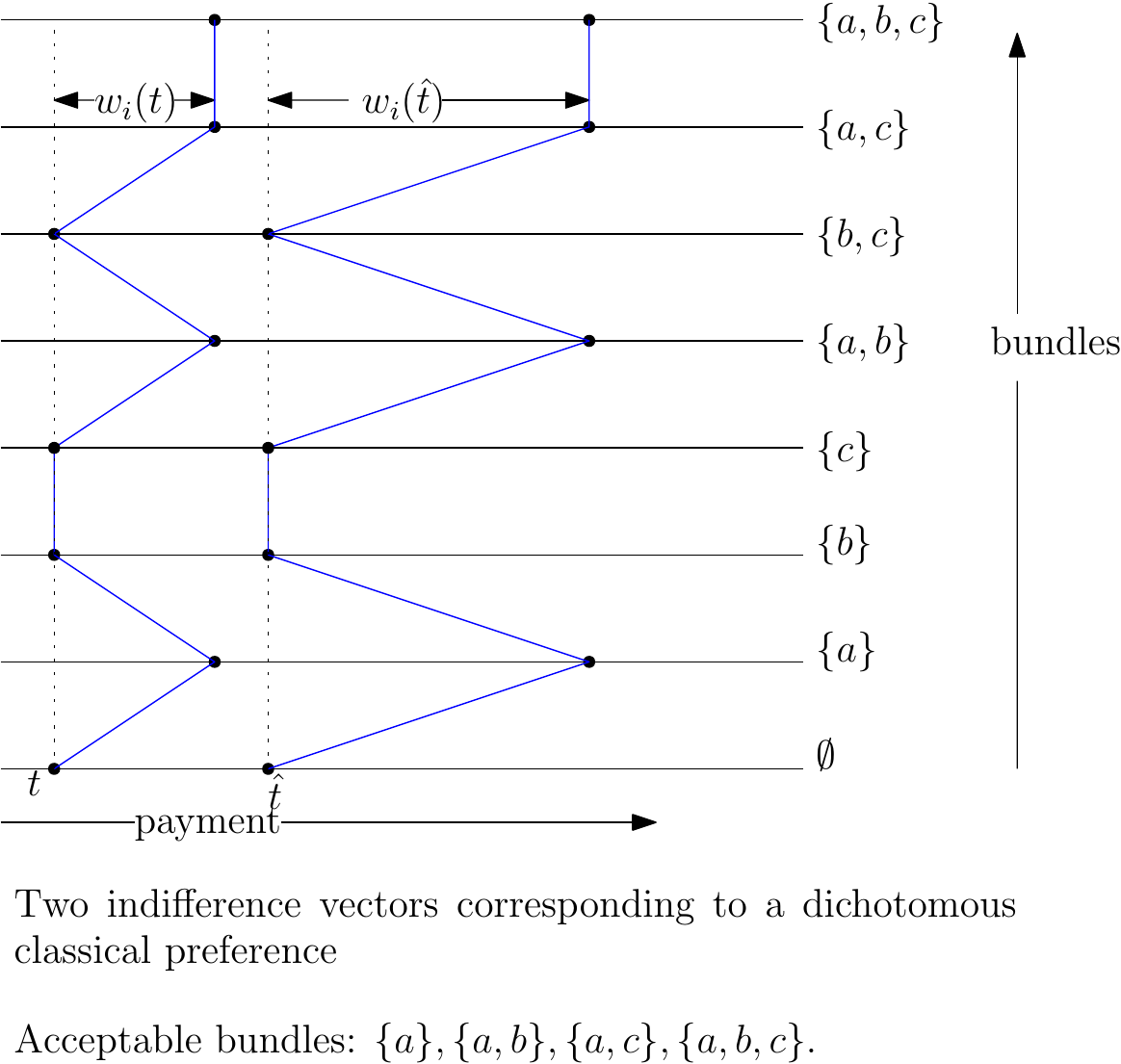}
\caption{A dichotomous preference}
\label{fig:dich}
\end{figure}

We will denote the domain of {\bf all} dichotomous preferences as $\mathcal{D}$, where each preference
in $\mathcal{D}$ for agent $i$ is described by a $w_i$ map and a collection of minimal bundles $\mathcal{S}^{min}_i$. A {\bf dichotomous domain} is any
subset of dichotomous preferences.

For some of our results, we will need a particular type of dichotomous preference.
\begin{defn}\label{def:single}
A dichotmous preference $R_i \equiv (\mathcal{S}_i^{min},w_i)$ is called a {\bf single-minded preference} if
$|\mathcal{S}_i^{min}|=1$.
\end{defn}
An agent having a single-minded dichotomous preference has {\em a unique} bundle of objects and all its supersets as acceptable bundles.
Let $\mathcal{D}^{single}$ denote the set of all single-minded preferences.
Single-minded preferences are well-studied in the algorithmic game theory literature~\citep{Lehmann02,Babaioff05,Babaioff09}.
They were also central in the recent analysis of US incentive auction~\citep{Milgrom19}. Our main negative result will
be for domains containing $\mathcal{D}^{single}$. Establishing a negative result
on domains containing $\mathcal{D}^{single}$ implies a negative result on domains containing $\mathcal{D}$ since $\mathcal{D}^{single} \subsetneq \mathcal{D}$.

Before concluding this section, we briefly discuss how dichotomous preferences
are similar to some other kinds of preferences in the literature. In the single object
model, the preferences are clearly dichotomous, where there is no uncertainty about
the acceptable bundles. Similarly, consider the unit demand preferences studied in
\citet{Demange85,Morimoto15}. A preference $R_i$ is a unit demand preference if for every
$S \in \mathcal{B}$ and every $t \in \mathbb{R}$, we have $WP(S,t;R_i)=\max_{a \in S}WP(\{a\},t;R_i)$.
Now, suppose the objects are {\em homogeneous} in the
following sense: $WP(\{a\},t;R_i)=WP(\{b\},t;R_i)$ for all $a,b \in M$ and for all $t \in \mathbb{R}$.
It is clear that a unit demand preference $R_i$ over homogeneous objects is a dichotomous preference,
where $\mathcal{S}^{min}_i$ consists of singleton bundles. If the objects are not homogeneous, the
unit demand preferences are not dichotmous since the willingness to pay of different objects may be different.

\section{The results}
\label{sec:results}

We describe our main results in this section.

\subsection{An impossibility result}
\label{sec:impos}

We start with our main negative result:
if the domain consists of {\em all} single-minded preferences, then
there is no desirable mechanism. This generalizes the intuition we demonstrated in the example
in Section \ref{sec:mot}.
\begin{theorem}[Impossibility]\label{theo:impos}
Suppose $\mathcal{T} \supseteq \mathcal{D}^{single}$ (i.e., the domain contains {\em all} single-minded preferences), $n \ge 3$, and $m \ge 2$.
Then, no desirable mechanism exists in $\mathcal{T}^n$.
\end{theorem}

The proof of this theorem and all other proofs are relegated to an appendix at the end.
The proof formalizes the sketch given in the example in Section \ref{sec:mot}.
The main idea of the proof is that if a desirable mechanism exists in $\mathcal{D}^{single}$,
it has to define outcomes at {\em all} single-minded preference profiles, which includes
an $n$-agent and $m$-object version of the preference profile discussed in Section \ref{sec:mot}.
The challenge is to show that any desirable mechanism at that profile must coincide with
the outcome of a {\em generalized} VCG mechanism (where agents pay their ``externalities").
Once this is shown, the rest of the proof is similar to the discussion in Section \ref{sec:mot}.

As discussed in the introduction, Theorem \ref{theo:impos} adds to a small list of papers that have established such negative results in other combinatorial
auction problems. Notice that the domain $\mathcal{T}$ may contain preferences that are not dichotomous or it may be equal to $\mathcal{D}$, the
set of all dichotomous preferences.

The conditions $m \ge 2$ and $n \ge 3$ are both necessary: if $m=1$, we
know that a desirable mechanism exists~\citep{Saitoh08}; if $n=2$, the mechanism that we propose next is desirable -- see
Proposition \ref{prop:vcg} and discussions after it.

\begin{defn}
The {\bf generalized Vickrey-Clarke-Groves mechanism with loser's payment $t_L$ (GVCG-$t_L$)}, denoted as $(f^{vcg,t_L},\mathbf{p}^{vcg,t_L})$, is defined as follows:
for every profile of preferences $R$,
\begin{align*}
f^{vcg,t_L}(R) &\in \arg \max_{A\in \mathcal{X} } \sum_{i\in N} WP(A_i,t_L;R_i) \\
p^{vcg,t_L}_i(R) &= t_L + \max_{A \in \mathcal{X}} \sum_{j \ne i} WP(A_j,t_L;R_j)- \sum_{j \ne i} WP(f^{vcg,t_L}_j(R),t_L;R_j).
\end{align*}
We refer to the GVCG-$0$ mechanism as the {\bf GVCG} mechanism.
\end{defn}

The GVCG class of mechanisms is a natural generalization of the VCG mechanism to our setting without quasilinearity.
Note that the current definition does not use anything about dichotomous preferences. It computes the ``externality'' of every agent with respect to a reference transfer
level $t_L$. This transfer level $t_L$ corresponds to
the payment by any agent who does not win any non-empty bundle of objects in the mechanism (such an agent has zero
externality).
The additional term $t_L$ in the payment expression ensures that when we use $t_L$ as the reference transfer level to
compute externalities, we maintain incentive compatibility in the dichotomous domain. In the quasilinear domain,
the reference transfer level does not matter as the willingness to pay does not change with reference transfer: $WP(S,t_L,R_i)=WP(S,0,R_i)$ for
each $S$, if $R_i$ is a quasilinear preference.

Theorem \ref{theo:impos} implies that the GVCG mechanism is not desirable.
Indeed, no GVCG mechanism can be DSIC in an arbitrary combinatorial auction domain without
quasilinearity. For instance, \citet{Morimoto15} show that there is a unique desirable
mechanism in the domain of ``unit-demand" (where agents have demand for at most one object) preferences,
and it is {\bf not} a GVCG mechanism.
We show that the GVCG mechanism is DSIC, individually rational, and satisfies no subsidy in {\em any} dichotomous preference domain.

\begin{prop}\label{prop:vcg}
Consider the GVCG-$t_L$ mechanism for some $t_L \in \mathbb{R}$, defined on an arbitrary dichotomous domain $\mathcal{T} \subseteq \mathcal{D}$. Then, the following are true.

\begin{enumerate}
\item The GVCG-$t_L$ mechanism is DSIC.
\item The GVCG-$t_L$ mechanism is individually rational if $t_L \le 0$.
\item The GVCG-$t_L$ mechanism satisfies individual rationality and no subsidy if $t_L=0$.
\item The GVCG-$t_L$ mechanism is Pareto efficient if $n=2$.
\item The GVCG-$t_L$ mechanism is not Pareto efficient if $n> 2, m > 1$, and $\mathcal{T} \supseteq \mathcal{D}^{single}$.
\end{enumerate}
\end{prop}

We explain below why the GVCG class of mechanisms are compatible with Pareto efficiency when $n=2$ but
not compatible when $n > 2$. For simplicity, we assume that preferences of agents are single-minded, i.e., the domain
is $\mathcal{D}^{single}$. We consider various cases. \\

\noindent {\sc One object ($m=1$).} It is well known that the GVCG mechanism is Pareto efficient if $m=1$~\citep{Saitoh08}.
Note that for $m=1$, every preference is single-minded.
The GVCG mechanism allocates the object
to an agent $k$ with the highest WP at $0$, i.e., $w_k(0) = \max_{i \in N}w_i(0)$. All agents except agent $k$
pay zero and agent $k$ pays $\max_{i \ne k}w_i(0)$. This outcome is always Pareto efficient. The main reason for this is
that there is {\em only one} object, and any new outcome can only give this object to one agent (may be the same or another agent).
Take any such outcome $z \equiv (z_1,\ldots,z_n)$ and assume for contradiction that it Pareto dominates the GVCG outcome.
If agent $k$ continues to get the object in $z_k$ also, her payment cannot be more than $\max_{i \ne k}w_i(0)$.
Further, payments of other agents cannot be more than zero. As a result, total payment cannot be more than $\max_{i \ne k}w_i(0)$.
Similarly, if any other agent $j \ne k$ receives the object in $z$, then her payment cannot be more than $w_j(0)$ (else,
she will prefer the GVCG outcome of getting nothing and paying zero). Further, in this case, since agent $k$ does not receive the object
in $z$, her payment will be non-positive. As a result, the total payment cannot be more than $\max_{i \ne k}w_i(0)$.
In fact, the total payment in $z$ in both the cases will be strictly less than the GVCG
payments if any agent strictly improves, which is a contradiction. \\

\noindent {\sc Two agents ($n = 2$) but arbitrary $m$.} Since preferences of agents are
single-minded, at every preference profile the acceptable bundles of each agent $i$ are supersets of some $S_i \in \mathcal{B}$. Since
there are two agents, we have only two cases to consider: (i) $S_1 \cap S_2 = \emptyset$
and (ii) $S_1 \cap S_2 \ne \emptyset$. Intuitively, in the first case, the two agents are not competing against each other.
Pareto efficiency requires us to allocate each agent $i \in \{1,2\}$ her acceptable bundle $S_i$. The GVCG mechanism charges zero payment to the agents. Clearly, this cannot be Pareto dominated.
In the second case, the two agents compete against each other like the single object case. This is because $S_1 \cap S_2 \ne \emptyset$ means
exactly one agent can be assigned an acceptable bundle. In fact the allocation and payment in the GVCG mechanism for this case
mirrors the single object case: the agent with the
higher WP at $0$ gets her acceptable bundle and pays the willingness to pay of the other agent.
The fact that this outcome cannot be Pareto dominated follows an argument similar to the $m=1$ case.
Summarizing, if there are two agents, independent
of the number of objects, the Pareto efficiency requirement is very similar to the single object case. Hence, the GVCG mechanism remains compatible
with Pareto efficiency. \\

\noindent {\sc More than two agents and more than one object} ($n > 2, m > 1$). With more than two agents and more than one object, the Pareto efficiency requirement
is no longer like the single object case. To understand, let us consider
Example \ref{ex:ex1} (see Table \ref{tab:negative}).
The GVCG mechanism
allocates objects $a$ and $b$ to agents $2$ and $3$ but charges them low payments ($1.9$ each).
This is akin to low payments
in the VCG mechanism as documented in \citet{Ausubel06}.\footnote{They point out that
when there are at least two objects and at least three agents, the VCG mechanism outcome
may not lie in the ``core" of the associated game if objects are complements. This in turn results
in low payments. The dichotomous preferences
exhibit extreme form of complementarity.} In our example, even though agent $1$
is not allocated any object, she has high enough willingness to pay for the bundle of objects -- with one object, if the payment
of the winning agent is low, then the willingness to pay of all losing agents is also low.
With negative income effect,
agents $2$ and $3$ feel ``wealthier" after getting the objects at low payments.
So, their ``willingness to sell" amount is low.
Hence, it is easier to compensate them.
With agent $1$ having a high enough willingness to pay ($3.9$), a Pareto improving trade is thus possible.
Such a Pareto improving trade is not possible if agents $2$ and $3$ have positive income effect preferences.
This is because with positive income effect, the ``willingness to sell" amount is higher than the willingness to pay.

\subsection{Positive income effect and possibility}
\label{sec:pos}

Proposition \ref{prop:vcg} and Theorem \ref{theo:impos} point out that the GVCG is not Pareto efficient in the entire
dichotomous domain. A closer look at the proof of Theorem \ref{theo:impos} (and Example \ref{ex:ex1})
reveals that the impossibility is driven by a particular kind of dichotomous preferences: the ones where the willingness
to pay of an agent increases with payment. We term such preferences {\em negative income effect}.

A standard definition of positive income effect will say that as income rises, a preferred bundle becomes ``more preferred".
We do not model income explicitly, but our preferences implicitly account for income. So, if payment decreases from $t$
to $t'$, the income level of the agent increases implicitly. As a result, she is willing to pay more for his acceptable bundles at
$t'$ than at $t$. Thus, positive income effect captures a reasonable (and standard) restriction on preferences of the agents.
\begin{defn}
\label{def:dichpos}
A dichotomous preference $R_i \equiv (w_i,\mathcal{S}_i)$ satisfies {\bf positive income effect} if for all $t > t'$, we have
$w_i(t) \le w_i(t').$

A dichotomous domain of preferences $\mathcal{T}$ satisfies positive income effect if every preference in $\mathcal{T}$ satisfies
positive income effect.
\end{defn}
As an illustration, the indifference vectors shown in Figure \ref{fig:dich} cannot be part of a dichotomous preference satisfying positive income effect -- we see
that $\hat{t} > t$ but $w_i(\hat{t}) > w_i(t)$. The preference $R_0$ in Example \ref{ex:ex1} also violated positive income effect.
A quasilinear preference (where $w_i(t)=w_i(t')$ for all $t,t'$) always satisfies positive income effect, and
the GVCG mechanism is known to be a desirable mechanism in this domain. We show below that the GVCG mechanism is Pareto efficient
if the domain contains preferences that satisfy positive income effect. Before stating the result, let us reconsider Example \ref{ex:ex1}
and see why the GVCG mechanism becomes desirable with positive income effect.

\begin{example}
\label{ex:ex11}
\end{example}
\begin{exam}
We revisit Example \ref{ex:ex1} but with an important difference: the preferences of agents 2 and 3 now satisfy positive income effect.
So, we have three agents $N=\{1,2,3\}$ and two objects $M=\{a,b\}$. As in Example \ref{ex:ex1}, agent $1$
has single-minded quasilinear preference $R_1$ with valuation $3.9$ on the unique acceptable bundle $\{a,b\}$.
All the bundles are acceptable bundles for agents 2 and 3. But their preference is now $\widehat{R}_0$ which satisfies positive income effect.
However, similar to Example \ref{ex:ex1}, we have $\widehat{w}(0)=2$.
Figure \ref{fig:pie_1} shows two indifference vectors of $\widehat{R}_0$. Since $\widehat{R}_0$ satisfies positive
income effect, we have $\widehat{w}(t) > \widehat{w}(0)$, where $t < 0$.

The GVCG outcome
does not change from Example \ref{ex:ex1} at this profile: agent $2$ gets object $a$ and agent $3$ gets object $b$ with
payments $p_1^{vcg}=0, p_2^{vcg}=p_3^{vcg}=1.9$. To Pareto dominate this outcome, we need to give both the objects to agent $1$.

\begin{figure}[!hbt]
\centering
\includegraphics[width=4in]{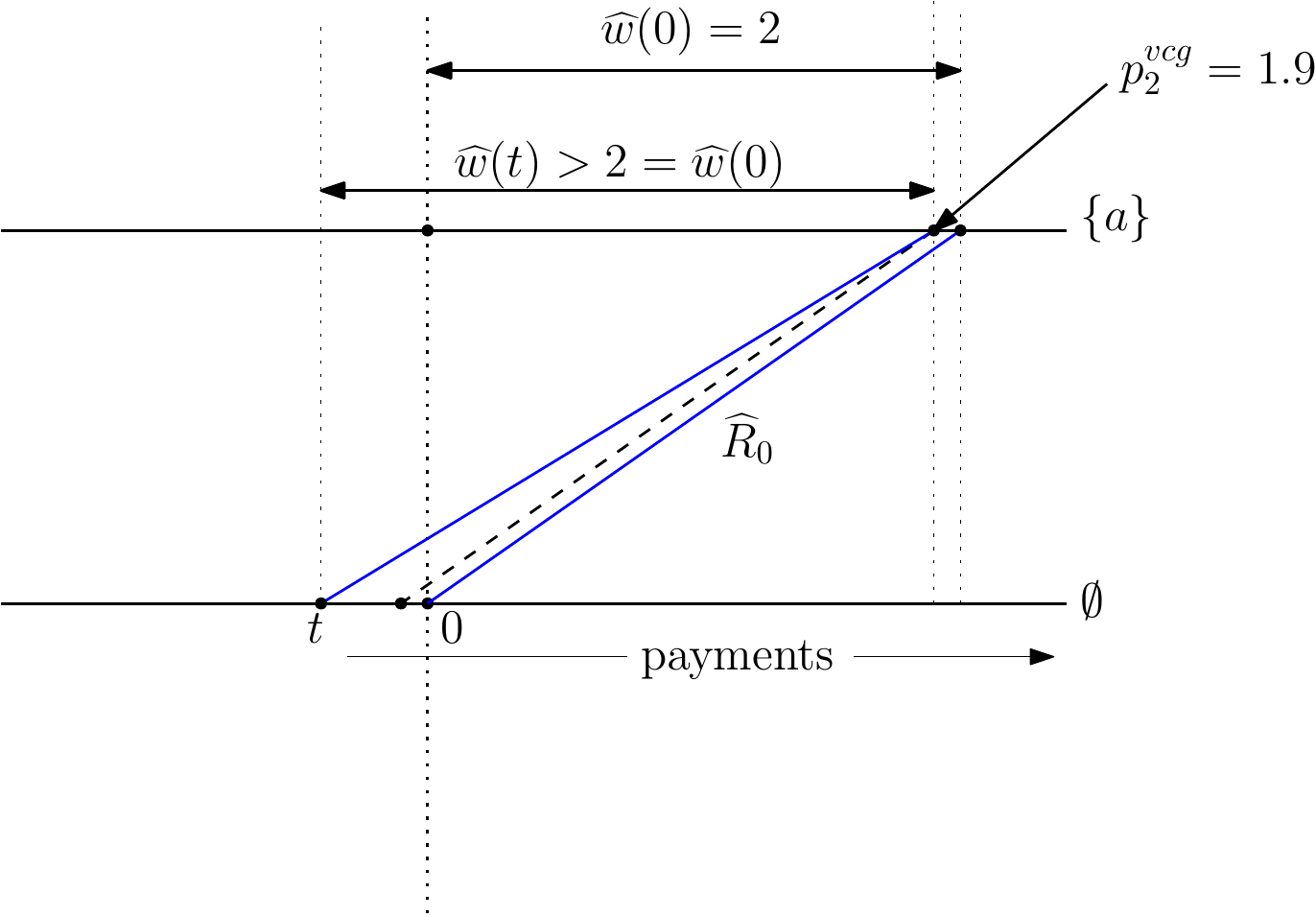}
\caption{Possibility with positive income effect}
\label{fig:pie_1}
\end{figure}

\begin{table}[!hbt]
  \centering
  \begin{tabular}{|c||c | c | c|}
    \hline
     & $\{a\}$ & $\{b\}$ & $\{a,b\}$ \\
     \hline
     \hline
    $WP(\cdot,0;R_1)$ & $0$ & $0$ & $3.9$ \\
    \hline
      $WP(\cdot,0;R_2=\widehat{R}_0)$ & $2$ & $2$ & $2$ \\
      $WP(\cdot,0;R_3=\widehat{R}_0)$ & $2$ & $2$ & $2$ \\
    \hline
    \hline
  \end{tabular}
\caption{A profiles of preferences with $M=\{a,b\}$, $N=\{1,2,3\}$.}
\label{tab:positive}
\end{table}

Now, the GVCG outcome to agent $2$ is $(\{a\},1.9)$ and, by Table \ref{tab:positive} (see Figure \ref{fig:pie_1} also), $(\{a\},2)~\widehat{I}_0~(\emptyset,0)$.
If $(\{a\},1.9)~\widehat{I}_0~(\emptyset,t)$, then by positive income effect $t < -0.1$. A pictorial description of the indifference
vectors of $\widehat{R}_0$ for these transfer amounts are shown in Figure \ref{fig:pie_1}.
This means that if agent $2$ is not given any object, the total compensation required for her
alone will be more than $0.1$. Since agent $3$ needs to be compensated too and the total revenue collected in the
VCG outcome is $3.8$, we need to charge more than $3.9$ to agent $1$ to Pareto dominate the VCG outcome. This is impossible since
the value of agent $1$ for both the objects is only $3.9$.
\end{exam}

The intuition in this example generalizes. The main idea is that the GVCG mechanism
allocates goods in a way that maximizes the collective willingness to pays (at zero) of the winning
bidders. IR implies that the winning bidders pay a price less than their willingness to pay
for their winning bundles. Thus, winning essentially makes the bidders feel ``wealthier".
Positive income effect then implies that their ``willingness to sell" after the auction exceeds the willingness
to pay before the auction. This rules out any Pareto improving trades\footnote{We are grateful
to an anonymous referee for this intuition. \citet{Baisa19jme} give similar intuition in a
single object auction model to establish a mapping between non-quasilinear and quasilinear
economies.}.

Our next result says that the impossibility in Theorem \ref{theo:impos} is overturned in any domain of dichotomous preferences satisfying positive income effect.
\begin{theorem}[Possibility]\label{theo:gvcg}
The GVCG mechanism is desirable on any dichotomous
domain satisfying positive income effect.
\end{theorem}

Theorem \ref{theo:gvcg} can be interpreted to be a generalization of the well-known result that the VCG mechanism
is desirable in the quasilinear domain. Indeed, we know that if the domain of preferences is the set of {\em all} quasilinear preferences, then
standard revenue equivalence result (which holds in the quasilinear domain) implies that the VCG mechanism is the {\em only} desirable
mechanism. Though we do not have a revenue equivalence result, we show below a similar uniqueness result of the GVCG mechanism.
For this, we first remind ourselves of the definition of a quasilinear preference.
A dichotomous preference $(w_i,\mathcal{S}_i)$ is {\bf quasilinear} if for every $t,t' \in \mathbb{R}$, we have $w_i(t)=w_i(t')$.
We denote by $\mathcal{D}^{QL}$ the set of {\bf all} dichotomous quasilinear preferences. This leads to a characterization
of the GVCG mechanism.

\begin{theorem}[Uniqueness]\label{theo:unique}
Suppose the domain of preferences $\mathcal{T}$ is a dichotomous domain satisfying positive income effect and contains $\mathcal{D}^{QL}$.
Let $(f,\mathbf{p})$ be a mechanism defined on $\mathcal{T}^n$. Then, the following statements are equivalent.
\begin{enumerate}
\item $(f,\mathbf{p})$ is a desirable mechanism.

\item $(f,\mathbf{p})$ is the GVCG mechanism.

\end{enumerate}
\end{theorem}

We reiterate that the GVCG is known to fail DSIC with non-quasilinear preferences
if agents demand multiple objects. So, Theorems \ref{theo:gvcg}
and \ref{theo:unique} show that under dichotomous classical preferences with
positive income effect, we recover the desirability of the GVCG mechanism.

\subsection{Tightness of results}
\label{sec:robust}

In this section, we investigate if the positive results in the previous sections continue to hold
if the domain includes (positive income effect) non-dichotomous preferences. In particular, we investigate the consequences of adding a non-dichotomous
preference satisfying positive income effect and some other reasonable properties (we precisely define them
later in the section).
Both these conditions are natural properties to impose on preferences.
Our results below can be summarized as follows: if we take the set of {\em all}
quasilinear dichotomous preferences and add {\em any} non-dichotomous preference satisfying the above two conditions, then
no desirable mechanism can exist in such a type space.
As corollaries, we uncover new type spaces where
no desirable mechanism can exist with non-quasilinear preferences, and establish the role of dichotomous
preferences in such type spaces. Before we formally state the result, we give an example
to show why we should expect such an impossibility result.

\begin{table}[!hbt]
  \centering
  \begin{tabular}{|c||c | c | c|}
    \hline
     & $\{a\}$ & $\{b\}$ & $\{a,b\}$ \\
     \hline
     \hline
    $WP(\cdot,0;R_1)$ & $0$ & $0$ & $5$ \\
    \hline
      $WP(\cdot,0;R_2=R_0)$ & $3$ & $4$ & $4$ \\
      $WP(\cdot,0;R_3=R_0)$ & $3$ & $4$ & $4$ \\
    \hline
    \hline
    $WP(\cdot,0;R'_2)$ & $0$ & $4$ & $4$ \\
    \hline
    \hline
  \end{tabular}
\caption{Two profiles of preferences with $M=\{a,b\}$, $N=\{1,2,3\}$.}
\label{tab:vcg_sp}
\end{table}

\begin{example}
\label{ex:ex2}
\end{example}
\begin{exam}
\noindent We consider an example with two object $M:=\{a,b\}$ and three agents $N:=\{1,2,3\}$.
We will require the following preferences of the agents.
The preference $R_1$ of agent $1$ is quasilinear and the corresponding values for
bundles of objects is shown in Table \ref{tab:vcg_sp}. It is clear that $R_1$ is a
single-minded preference.
We have two preferences of agent $2$: $R_2=R_0$ and $R'_2$.
Preference $R_0$ is not quasilinear, but
it satisfies positive income effect (decreasing prices by the same amount of two indifferent consumption bundles
lead the agents to strictly prefer the costlier object): $(\{b\},4)~I_0~(\{a\},3)$ and $(\{b\},2)~P_0~(\{a\},1)$.
This is shown in Figure \ref{fig:vcg_sp}, where we show some indifference vectors of $R_0$. Note that
the other indifference vectors of $R_0$ can be constructed such that it satisfies the unit demand property and positive
income effect. Preference $R'_2$ is a quasilinear single-minded preference with $\{b\}$ and $\{a,b\}$ as
acceptable bundles and value $4$. Finally, preference $R_3$ of agent $3$ is also $R_0$.

\begin{figure}[!hbt]
\centering
\includegraphics[height=3in]{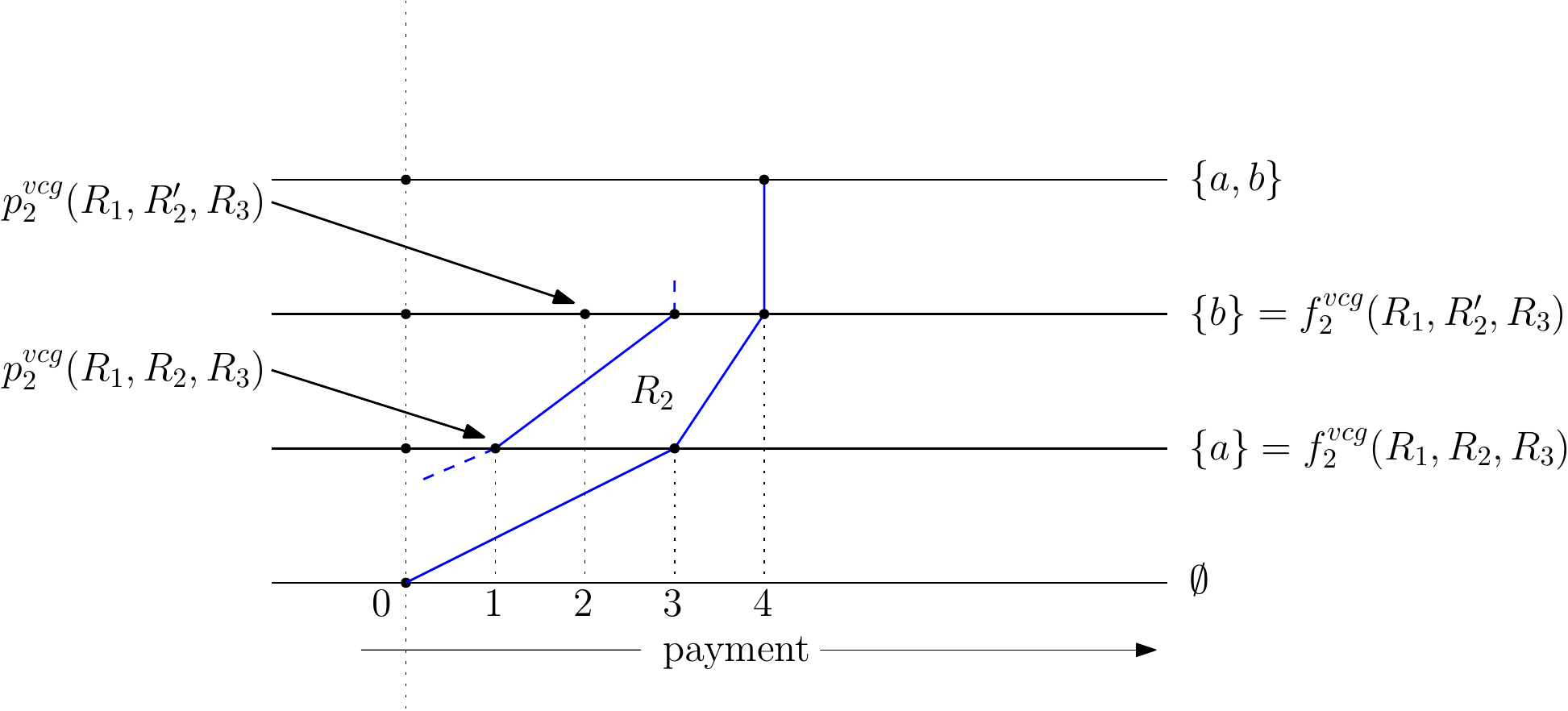}
\caption{Positive income effect preference of agents $2$ and $3$.}
\label{fig:vcg_sp}
\end{figure}

We argue that the GVCG mechanism containing all quasilinear dichotomous preferences and $R_0$ is {\bf not DSIC}.
So, our domain is $\mathcal{T}=\mathcal{D}^{QL} \cup \{R_0\}$. We will look at two preference profiles: $(R_1,R_2,R_3)$
and $(R_1,R'_2,R_3)$. At the preference profile $(R_1,R_2,R_3)$, agents 2 and 3 should get objects from
$\{a,b\}$ according to
GVCG. Since they have identical preferences, we break the tie by giving object $a$ to agent $2$ and object $b$ to agent $3$:
$f_1^{vcg}(R_1,R_2,R_3)=\{a\}, f^{vcg}_2(R_1,R_2,R_3)=\{b\}$.\footnote{
The example can be modified to work if the tie is broken by giving object $b$ to agent $2$ and object $a$ to agent $3$.}
The payment of agent $2$ is $p^{vcg}_2(R_1,R_2,R_3)=1$.

Now, consider the preference profile $(R_1,R'_2,R_3)$. Here, since agent $2$ has only $\{b\}$ and $\{a,b\}$ in her acceptable bundle, her
GVCG outcome changes: $f_2^{vcg}(R_1,R'_2,R_3) = \{b\}$ and $p_2^{vcg}(R_1,R'_2,R_3)=2$.
In other words, the externality of agent $2$ changes from $1$ at preference profile $(R_1,R_2,R_3)$ to $2$ at $(R_1,R'_2,R_3)$.

If $R_2$ was a quasilinear preference, then
agent $2$ would have been indifferent between $(\{a\},1)$ and $(\{b\},2)$. But since $R_2 = R_0$ satisfies positive income effect (see Figure \ref{fig:vcg_sp}),
$(\{b\},2)~P_2~(\{a\},1)$. This shows that with positive income effect, agent $2$ can manipulate in the GVCG mechanism in this domain.

This is a general problem. We formalize this in Theorem \ref{theo:robustu}. We show in the proof of Theorem \ref{theo:robustu}
that any desirable mechanism in such a domain must have the GVCG outcomes at these profiles,
and this will lead to manipulation by the agent having positive income effect.

It is crucial that $WP(\{a\},0;R_0) < WP(\{b\},0;R_0)$ for this manipulation to happen in this example. If $WP(\{a\},0;R_0) = WP(\{b\},0;R_0)=4$, then
$R_0$ can be a dichotomous preference (i.e., besides the indifference vector shown in Table \ref{tab:vcg_sp},
we can construct other indifference vectors such that it is a dichotomous preference). We know that the
GVCG mechanism is DSIC in such domains. Indeed, in that case, the externality of agent $2$ remains unchanged across profiles
$(R_1,R_2,R_3)$ and $(R_1,R'_2,R_3)$. In other words, we have $p^{vcg}_2(R_1,R_2,R_3)=p^{vcg}_2(R_1,R'_2,R_3)=1$. So, no manipulation
is possible by agent $2$ across these two preference profiles.\footnote{This is true even if this preference does not satisfy
positive income effect.}

\end{exam}

We formalize the intuition in Example \ref{ex:ex2} now.
We consider a preference where an agent can demand multiple heterogeneous objects.
We require that at least two objects are heterogeneous in the following sense.
\begin{defn}
A preference $R_i$ satisfies {\bf heterogenous} demand if there exists $a,b \in M$,
$$WP(\{a\},0;R_i) \ne WP(\{b\},0;R_i).$$
\end{defn}
Heterogeneous demand requires that for {\em some} pair of objects, the WP at $0$ must be different for them.
If objects are not the same (i.e., not homogeneous), then we should expect this condition to hold. We can provide an analogous tightness result if objects are homogeneous.\footnote{The result is
available on request.}

Besides the heterogeneous demand, we will impose two natural conditions on preferences. The first condition is a mild
form of substitutability condition.
\begin{defn}
A preference $R_i$ satisfies {\bf strict decreasing marginal WP} if for every $a,b \in M$,
$$WP(\{a\},0;R_i)  + WP(\{b\},0;R_i) > WP(\{a,b\},0;R_i).$$
\end{defn}
Strict decreasing marginal WP requires a minimal degree of submodularity: the marginal increase in WP (at $0$) by adding $\{a\}$
to $\{b\}$ is less than adding $\{a\}$ to $\emptyset$. Notice that this substitutability requirement is {\em only} for bundles
of size two. Hence, larger bundles may exhibit complementarity or substitutability.
Because of free disposal, for every $a,b \in M$, we have $$WP(\{a,b\},0;R_i) \ge \max(WP(\{a\},0;R_i),WP(\{b\},0;R_i)).$$
Hence, strict decreasing marginal WP implies that $WP(\{a\},0;R_i) > 0$ and $WP(\{b\},0;R_i) > 0$, i.e., each object is a {\em good} in a
weak sense (getting an object is preferred to getting nothing at payment $0$).

We point out that unit demand preferences (studied in \citep{Demange85,Morimoto15})
satisfy strict decreasing marginal WP. A preference $R_i$ is called a {\bf unit demand} preference if for every $S$, $$WP(S,t;R_i)=\max_{a \in S}WP(\{a\},t;R_i)~\forall~t \in \mathbb{R}_+.$$
If $R_i$ is a unit demand preference and objects are {\em goods}, then it satisfies strict decreasing marginal WP. To see this, call every object $a \in M$ {\bf a real good} if
 $WP(\{a\},0;R_i) > 0$ at every $R_i$. If every object is a real good, then for every $a, b \in M$, we see that
$$WP(\{a\},0;R_i)+WP(\{b\},0;R_i) > \max_{x \in \{a,b\}}WP(\{x\},0;R_i) = WP(\{a,b\},0;R_i).$$

Besides the strict decreasing marginal WP condition, we will also be requiring strict positive income effect, but {\em only} for
singleton bundles.
\begin{defn}
A classical preference $R_i$ satisfies {\bf strict positive income effect} if for every $a, b \in M$ and
for every $t, t'$ with $t' > t$, the following holds for every $\delta > 0$:
$$\Big[(\{b\},t')~I_i~(\{a\},t)\Big]~\Rightarrow~\Big[(\{b\},t'-\delta)~P_i~(\{a\},t-\delta)\Big].$$
\end{defn}
This definition of strict positive income effect requires that if two objects are indifferent then decreasing their
prices by the same amount makes the higher priced (lower income) object better. This is a generalization of the definition of
positive income effect we had introduced for dichotomous preferences in Definition \ref{def:dichpos}, but only
restricted to singleton bundles.\footnote{An alternate definition along the lines of Definition \ref{def:dichpos}
using willingness to pay map is also possible. It will require {\em decreasing differences
of willingness to pay}. Formally, a preference $R_i$ satisfies
strict positive income effect if for every $t' > t$ and for every $a, b \in M$, we have
$WP(\{a\},t';R_i) > WP(\{b\},t';R_i)$ implies $WP(\{a\},t';R_i) - WP(\{b\},t';R_i) < WP(\{a\},t;R_i) - WP(\{b\},t;R_i)$.}
This means that for larger bundles, we do not require positive income effect to hold.

We are ready to state the main tightness result with heterogeneous objects.
\begin{theorem}\label{theo:robustu}
Suppose $n \ge 4, m \ge 2$. Let $R_0$ be a heterogeneous demand preference satisfying strict positive income effect and
strict decreasing marginal WP.
Consider any domain $\mathcal{T}$ containing $\mathcal{D}^{QL} \cup \{R_0\}$.
Then, no desirable mechanism exists in $\mathcal{T}^n$.
\end{theorem}

We make a quick remark about the statement of Theorem \ref{theo:robustu}. \\

\noindent {\sc Remark 1.}\begin{exam}
Though Theorem \ref{theo:robustu} requires $n \ge 4$, a careful look at its proof reveals that
we only need $n \ge 4$ if $m > 2$. If there are only two objects, the impossibility result in Theorem \ref{theo:robustu}
holds with $n \ge 3$. This was shown in Example \ref{ex:ex2} also.

The basic idea of the proof of Theorem \ref{theo:robustu} is similar to Example \ref{ex:ex2}.
With more than two object ($m > 2$), we will need at least four agents. The reason is
slightly delicate. Notice that $R_0$ in the statement of Theorem \ref{theo:robustu}
is an arbitrary preference. As in Example \ref{ex:ex2}, the proof ensures that three agents
compete for two objects, say $\{a,b\}$, out of which two agents have $R_0$ as their preference. With more
than two objects, we need a way to ensure that $\{a,b\}$ are allocated among these three agents.
In the absence of a fourth agent, it is not possible to ensure that the two agents having $R_0$
preference are not assigned objects outside of $\{a,b\}$. A fourth agent having arbitrarily large willingness
to pay for the bundle $M \setminus \{a,b\}$ ensures that.

We do not know if the impossibility result holds for $n=2$ or $n=3$ when $m > 2$, but we conjecture that it does not.
\end{exam}

Unlike the negative result in Theorem \ref{theo:impos}, Theorem \ref{theo:robustu} does not require the existence
of negative income effect dichotomous preferences. It requires the domain to contain the set of quasilinear dichotomous preferences
and one heterogeneous demand preference satisfying some reasonable conditions.
This negative result parallels a result of \citet{Kazumura16} who show that adding {\em any} multi-demand
preference to a class of {\em rich} unit demand preference gives rise to a similar impossibility.
As was explained in Example \ref{ex:ex2}, our proof exploits the fact that
any desirable mechanism must coincide with the GVCG mechanism in the positive income effect dichotomous domain,
and adding any strictly positive income effect preference
to the domain leads to manipulation. In the case of \citet{Kazumura16}, they add an arbitrary multi-demand preference (which may or may not satisfy positive income effect) to a domain of unit demand
preferences, where the GVCG mechanism is {\em not} desirable. So, neither of the results imply
the other and the proof strategies are different.

We now spell out an exact implication of Theorem \ref{theo:robustu} in a corollary below.
Let $\mathcal{D}^+$ be
the set of all {\em positive income effect} dichotomous preferences (note that $\mathcal{D}^{QL} \subsetneq
\mathcal{D}^+$) and $\mathcal{U}^{+}$
be the set of all heterogeneous unit demand preferences satisfying positive income effect (as argued earlier, unit demand preferences
satisfy strict decreasing marginal WP). Then, the following corollary is immediate from Theorem \ref{theo:robustu}.
\begin{cor}\label{cor:imp}
Suppose $\mathcal{T}=\mathcal{D}^+ \cup \mathcal{U}^+$. Then, no desirable mechanism  exists on $\mathcal{T}^n$.
\end{cor}

Theorem \ref{theo:unique} shows that the GVCG mechanism is the unique desirable mechanism on $\mathcal{D}^+$.
Similarly, \citet{Demange85} have shown that a desirable mechanism exists in $\mathcal{U}^+$. This mechanism
is called the {\em minimum Walrasian equilibrium price mechanism} and collapses to the VCG mechanism if
preferences are quasilinear. Corollary \ref{cor:imp} says that we lose these possibility results if we consider
the unions of these two type spaces.

\section{Related Literature}
\label{sec:lit}

The quasilinearity assumption is at the heart of mechanism design literature with payments.
Our formulation of classical preferences was studied in the context of single object auction
by \citet{Saitoh08}, who proposed the generalized VCG mechanism and axiomatized it for that setting. Other such
axiomatizations include \citet{Sakai08,Sakai13}. As discussed,
\citet{Demange85} had shown that a mechanism different from the generalized VCG mechanism is desirable
when multiple heterogeneous objects are sold to agents with unit demand. Characterizations of this
mechanism have been given in \citet{Morimoto15}, \citet{Serizawa16} and \citet{Kazumura18b}. However, impossibility results for the existence
of a desirable mechanism were shown (a) by \citet{Kazumura16} for multi-object auctions with multi-demand agents
and (b) by \citet{Baisa17} for multiple homogeneous object model with multi-demand agents.
Social choice problems with payments are studied with particular form of non-quasilinear preferences
in \citet{Ma16,Ma18}. These papers establish dictatorship results in this setting with non-quasilinear preferences.

\citet{Baisa16} considers non-quasilinear preferences with randomization in a single object auction
environment. He proposes a randomized mechanism and establishes strategic properties of this mechanism. \citet{Dastidar15}
considers a model where agents have same utility function but models income explicitly to allow for different incomes. He considers
equilibria of standard auctions. \citet{Larry18}
discuss an implementation duality without quasilinear preferences and apply it to matching and adverse selection problems.
\citet{Kazumura18} discuss monotonicity based characterization of DSIC mechanisms
in domains which admit non-quasilinear preferences. \citet{Baisa19jme} discuss a model of
single object allocation when bidders have interdependent values and non-quasilinear preferences with positive
income effects. They give necessary and sufficient conditions for the existence of an ex-post implementable
and Pareto efficient mechanism in two settings: (i) where the auctioneer is the seller; and (ii) the procurement
setting, where the auctioneer is the buyer. In the former setting, their condition requires existence of
an ex-post implementable and Pareto efficient mechanism in a corresponding quasilinear economy. In the latter setting,
they show an impossibility result if the level of interdependence is strong.

The literature on auction design with budget constrained bidders models
budget constraint such that if an agent has to pay more than budget, then his utility is minus infinity. This introduces non-quasilinear
utility functions but it does not fit our model because of the hard budget constraint. For the multi-unit auction with such budget-constrained
agents, \citet{Lavi12} establish that no desirable mechanism can exist -- see an extension of this result
in \citet{Dobzinski12}. They prove this result by considering two bidders each with publicly known budgets
and two units. Their result shows an impossibility similar to ours as long as the public budgets of the bidders
are not equal. Their paper also allows complementary preferences but not of the extreme form seen with dichotomous preferences.

For combinatorial auctions with single-minded and quasilinear preferences,
\citet{Le18} shows that these impossibilities with budget-constrained agents can be overcome in a {\em generic} sense --
he defines a ``truncated" VCG mechanism and shows that it is desirable {\em almost everywhere}.

There is a literature in algorithmic mechanism design on combinatorial auctions with
quasilinear but ``single-minded" preferences. Apart from practical significance,
the problem is of interest because computing a VCG outcome is computationally challenging but
various ``approximately" desirable mechanisms which are computationally tractable can be constructed~\citep{Babaioff05,Babaioff09, Lehmann02,Milgrom19}.
\citet{Rastegari11} show that in this model, the revenue from the VCG mechanism (and any DSIC mechanism) may
not satisfy monotonicity, i.e., adding an agent may {\em decrease} revenue. Our paper adds to this literature by
illustrating the implications of non-quasilinear preferences.

\appendix

\section{Proofs}

\subsection{Proof of Theorem \ref{theo:impos}}

The proof extends the intuition in Example \ref{ex:ex1}. \\

\begin{proof}
We start by providing two useful lemmas.
\begin{lemma}\label{lem:ir}
Suppose $(f,\mathbf{p})$ is an individually rational mechanism
satisfying no subsidy. Then for every agent $i \in N$ and every $R \in \mathcal{T}^n$, we have $p_i(R)=0$
if $f_i(R) \notin \mathcal{S}_i$.
\end{lemma}
\begin{proof}
Suppose $R$ is a profile such that $f_i(R) \notin \mathcal{S}_i$ for agent $i$. By individual rationality,
$(f_i(R),p_i(R))~R_i~(\emptyset,0)$. But $f_i(R) \notin \mathcal{S}_i$ implies that
$(\emptyset,p_i(R))~I_i~(f_i(R),p_i(R))~R_i~(\emptyset,0)$. Hence, $p_i(R) \le 0$.
But no subsidy implies that $p_i(R)=0$.
\end{proof}

\begin{lemma}\label{lem:ir2}
Suppose $(f,\mathbf{p})$ is an individually rational mechanism
satisfying no subsidy. Then for every agent $i \in N$ and every $R \in \mathcal{T}^n$, we have
$0 \le p_i(R) \le WP(f_i(R),0;R_i).$
\end{lemma}
\begin{proof}
If $f_i(R) \notin \mathcal{S}_i$, then the claim follows from Lemma \ref{lem:ir}. Suppose $f_i(R) \in \mathcal{S}_i$.
By individual rationality, $(f_i(R),p_i(R))~R_i~(\emptyset,0)~I_i~(f_i(R),WP(f_i(R),0;R_i))$.
This implies that $p_i(R) \le WP(f_i(R),0;R_i)$. No subsidy implies that $p_i(R) \ge 0$.
\end{proof}

Consider any three non-empty bundles $S, S_1, S_2$ such that
$S=S_1 \cup S_2$ and $S_1 \cap S_2 = \emptyset$. Consider
a profile of single-minded preferences $R^* \in \big(\mathcal{D}^{single}\big)^n$ as follows.
Since all the agents have dichotomous preferences, to
describe any agent $i$'s preference, we describe the {\em minimal} acceptable bundles $\mathcal{S}^{min}_i$ (i.e., the
set of acceptable bundles $\mathcal{S}_i$ are derived by taking supersets of each element in $\mathcal{S}^{min}_i$) and
the willingness to pay map $w_i$.
Preference $R^*_1$ of agent $1$ is quasilinear:
$$\mathcal{S}^{min}_1 = \{S\}, w_1(t) = 3.9~\forall~t \in \mathbb{R}.$$
Preference $R^*_2$ of agent $2$ is:
$$\mathcal{S}^{min}_2=\{S_1\}, w_2(t) =2+3t~\forall~t>- \dfrac{1}{2} \textrm{ and } w_2(t) =\dfrac{1}{2} \textrm { otherwise}$$
Preference $R^*_3$ of agent $3$ is:
$$\mathcal{S}^{min}_3=\{S_2\}, w_3(t) =2+3t~\forall~t>- \dfrac{1}{2} \textrm{ and } w_3(t) =\dfrac{1}{2} \textrm { otherwise}$$
Preference $R^*_i$ of each agent $i \notin \{1,2,3\}$ is quasilinear:
$$\mathcal{S}^{min}_i = \{S\}, w_i(t) = \epsilon~\forall~t \in \mathbb{R},$$
where $\epsilon > 0$ but very close to zero.

Assume for contradiction that there exists a DSIC, Pareto efficient, individually rational
mechanism $(f,\mathbf{p})$ satisfying no subsidy.
We now do the proof in several steps. \\

\noindent {\sc Step 1.} In this step, we show that at every preference profile
$R$ with $R_i=R^*_i$ for all $i \notin \{2,3\}$, we must have $S \nsubseteq f_i(R)$ if $i \notin \{1,2,3\}$.
We know that $\mathcal{S}^{min}_i=\{S\}$ for all $i \notin \{2,3\}$.
Assume for contradiction $S \subseteq f_k(R)$ for some $k \notin \{1,2,3\}$.
Then, $S \nsubseteq f_1(R)$. By Lemma \ref{lem:ir}, $p_1(R)=0$.
Consider the following outcome:
$$Z_1 = (S,\epsilon), Z_k= (\emptyset, p_k(R) - \epsilon), Z_j = (f_j(R),p_j(R))~\forall~j \notin \{1,k\}.$$
Since preferences of agent $1$ and agent $k$ are quasilinear (note that $R_1=R^*_1$ and $R_k=R^*_k$)
and $\epsilon$ is very close to zero, we have
$$Z_1~P_1~(f_1(R),p_1(R)=0), Z_k~I_k~(f_k(R),p_k(R)), Z_j~I_j~(f_j(R),p_j(R))~\forall~j \notin \{1,k\}.$$
Also, the sum of payments in the outcome vector $Z \equiv (Z_1,\ldots,Z_n)$ is
$\sum_{i \in N}p_i(R).$ This contradicts Pareto efficiency of $(f,\mathbf{p})$. \\

\noindent {\sc Step 2.} Fix a preference $\hat{R}_2$ of agent 2 such that
$\mathcal{\hat{S}}^{min}_2 = \{S_1\}$ and $\hat{w}_2(0) > 1.9$. We show that
at preference profile $\hat{R}=(\hat{R}_2,R^*_{-2})$, $S \nsubseteq f_1(\hat{R})$.
Suppose $S \subseteq f_1(\hat{R})$. Then, $S_1 \nsubseteq f_2(\hat{R})$ and $S_2 \nsubseteq f_3(\hat{R})$.
By Lemma \ref{lem:ir}, $p_2(\hat{R})=0,p_3(\hat{R})=0$.
Consider a new outcome vector:
$$Z_1= (\emptyset,p_1(\hat{R})-3.9), Z_2 = (S_1,\hat{w}_2(0)), Z_3= (S_2,w_3(0)), Z_j=(f_j(\hat{R}),p_j(\hat{R}))~\forall~j \notin \{1,2,3\}.$$
By quasilinearity of $R^*_1$, we get $Z_1~I^*_1~(f_1(\hat{R}),p_1(\hat{R}))$. By definition,
$$Z_2~\hat{I}_2~(\emptyset,0)~\hat{I}_2~(f_2(\hat{R}),p_2(\hat{R})).$$ Similarly, $Z_3~I^*_3~(f_3(\hat{R}),p_3(\hat{R}))$.
Further, the sum of payments in
the outcome vector $Z$ is
$$p_1(\hat{R})-3.9 + \hat{w}_2(0) + w_3(0) + \sum_{j \notin \{1,2,3\}}p_j(\hat{R}) > \sum_{j \in N}p_j(\hat{R}),$$
where the inequality used the fact that $p_2(\hat{R})=p_3(\hat{R})=0$ and $\hat{w}_2(0) > 1.9, w_3(0)=2$.
This contradicts Pareto efficiency of $(f,\mathbf{p})$. \\

\noindent {\sc Step 3.} Fix any quasilinear preference $\hat{R}_2$ of agent $2$ such that
$\mathcal{\hat{S}}^{min}_2 = \{S_1\}$ and $\hat{w}_2(t) = 1.9 - \delta$, where $\delta \in (0,1.9)$. We show that
at preference profile $\hat{R}=(\hat{R}_2,R^*_{-2})$, we must have
$S \subseteq f_1(\hat{R})$. If not, then by Step 1 and by Pareto efficiency,
$S_1 \subseteq f_2(\hat{R})$ and $S_2 \subseteq f_3(\hat{R})$. Now, consider the following outcome $Z'$:
$$Z'_1=(S,3.9),~Z'_2=\Big(\emptyset,p_2(\hat{R})-(1.9-\frac{\delta}{2})\Big),~Z'_3=(\emptyset,p_3(\hat{R})-2),$$
$$Z'_j=(f_j(\hat{R}),p_j(\hat{R}))~\forall~j \notin \{1,2,3\}.$$

Note that by Lemma \ref{lem:ir}, $p_1(\hat{R})=0$. Hence, using quasilinearity of $R^*_1$, we get $(f_1(\hat{R}),p_1(\hat{R})=0)~I^*_1~(S,3.9)$.
Similarly, by quasilinearity of $\hat{R}_2$, we get $Z'_2~\hat{P}_2~(f_2(\hat{R}),p_2(\hat{R}))$. Also, the sum of payments
in outcome $Z'$ is
$$3.9 + p_2(\hat{R}) - (1.9 - \frac{\delta}{2}) + p_3(\hat{R}) - 2 + \sum_{j \notin \{1,2,3\}}p_j(\hat{R}) = \sum_{i \in N}p_i(\hat{R}) + \frac{\delta}{2} > \sum_{i \in N}p_i(\hat{R}),$$
where we used the fact that $p_1(\hat{R})=0$.

We now prove that $(\emptyset,p_3(\hat{R})-2)~R^*_3~(f_3(\hat{R}),p_3(\hat{R}))$. For this, let $t=p_3(\hat{R})-2$. Note that $w(t) \le 2$ follows from the definition of $w$ and the fact that $t \le 0$ by Lemma \ref{lem:ir2}. This implies
$(\emptyset,t)~R^*_3~(f_3(\hat{R}),t+2) $ i.e. $(\emptyset,p_3(\hat{R})-2)~R^*_3~(f_3(\hat{R}),p_3(\hat{R}))$

Hence, we get a contradiction to Pareto efficiency. \\

\noindent {\sc Step 4.} In this step, we show that at preference profile $R^*$,
$$S_1 \subseteq f_2(R^*), S_2 \subseteq f_3(R^*),$$
and $$p_2(R^*)=p_3(R^*)=1.9.$$
Since $w_2(0) = 2$ in preference $R^*_2$, by Step 2, $S \nsubseteq f_1(R^*)$. By Step 1,
$S \nsubseteq f_i(R^*)$ for all $i \notin \{1,2,3\}$. By Pareto efficiency, it must be
$$S_1 \subseteq f_2(R^*), S_2 \subseteq f_3(R^*).$$

Now, assume for contradiction $p_2(R^*) > 1.9$.
Fix a preference $\hat{R}_2$ of agent 2 such that
$\mathcal{\hat{S}}^{min}_2 = \{S_1\}$ and $p_2(R^*) > \hat{w}_2(0) > 1.9$. By Step 2, $S_1 \subseteq f_2(\hat{R}_2,R^*_{-2})$.
By DSIC, $p_2(R^*)=p_2(\hat{R}_2,R^*_{-2})$. Hence, $p_2(\hat{R}_2,R^*_{-2}) > \hat{w}_2(0)$.
This is a contradiction to Lemma \ref{lem:ir2}.

Finally, assume for contradiction $p_2(R^*) < 1.9$. Then, consider any quasilinear preference $\hat{R}_2$ of agent $2$ such that
$\mathcal{\hat{S}}^{min}_2 = \{S_1\}$ and $p_2(R^*) < \hat{w}_2(0) < 1.9$. By Step 3, $S_1 \nsubseteq f_2(\hat{R}_2,R^*_{-2})$ and by Lemma \ref{lem:ir},
$p_2(\hat{R}_2,R^*_{-2})=0$. But by
reporting $R^*_2$, agent $2$ gets $S_1$ at a payment less than $\hat{w}_2(0)$. By quasilinearity of $\hat{R}_2$ and the fact that
$S_1 \nsubseteq f_2(\hat{R}_2,R^*_{-2})$, she prefers this outcome
to outcome $(f_2(\hat{R}_2,R^*_{-2}),0)$, which is a contradiction to DSIC.

This concludes the proof that $p_2(R^*)=1.9$. A similar argument establishes (with Steps 2 and 3 applied to agent 3)
that $p_3(R^*)=1.9$. \\

\noindent {\sc Step 5.} We now complete the proof. By Step 4, we know that the outcome
at preference profile $R^*$ satisfies:
$$S \nsubseteq f_1(R^*), S_1 \subseteq f_2(R^*), S_2 \subseteq f_3(R^*),$$
$$p_1(R^*)=0,~p_2(R^*)=p_3(R^*)=1.9.$$
Note that by Lemma \ref{lem:ir}, $p_j(R^*)=0$ for all $j \notin \{1,2,3\}$.\\
Now, consider the following outcome: $Z'_j=(f_j(R^*),p_j(R^*))$ for all $j \notin \{1,2,3\}$ and
$$Z'_1=(S,3.9),Z'_2=(\{\emptyset\},-0.025),~Z'_3=(\{\emptyset\},-0.025).$$
Note that sum of payments in $Z'$ is $3.85 > p_2(R^*)+p_3(R^*)=3.8$.\\
Agent $1$ is indifferent between $Z'_1$ and $(f_1(R^*),p_1(R^*))$. Agents 2 and 3 are also indifferent between $Z'_i$ and $(f_i(R^*),p_i(R^*))$. This follows from the fact
that
$(-0.025)+w_2(-0.025)=(-0.025) + w_3(-0.025)=  1.9.$

This contradicts Pareto efficiency.
\end{proof}

\subsection{Proof of Proposition \ref{prop:vcg}}

\begin{proof}
Fix a dichotomous domain $\mathcal{T}$. For some $t_L \in \mathbb{R}$, consider the GVCG-$t_L$ mechanism and denote it
as $(f,\mathbf{p}) \equiv (f^{vcg,t_L},\mathbf{p}^{vcg,t_L})$.
We prove the following claim first.
\begin{claim}\label{cl:loser}
For every agent $i \in N$ and for every profile of preferences $R \in \mathcal{T}^n$, the following hold:
\begin{align}
(f_i(R),p_i(R))~&R_i~(\emptyset,t_L), \label{eq:tl} \\
p_i(R)&=t_L~\qquad~\textrm{if}~f_i(R) \notin \mathcal{S}_i, \label{eq:tl2}
\end{align}
where $\mathcal{S}_i$ is the acceptable set of bundles of agent $i$ at $R_i$.
\end{claim}
\begin{proof}
The following inequalities follow straightforwardly.
\begin{align*}
\max_{A \in \mathcal{X} } \sum_{j\in N} WP(A_j,t_L;R_j) &\geq \max_{A \in \mathcal{X}} \sum_{j\ne i} WP(A_j,t_L;R_j) \\
\Rightarrow \sum_{j\in N} WP(f_j(R),t_L;R_j) &\geq \max_{A \in \mathcal{X}} \sum_{j\ne i} WP(A_j,t_L;R_j) \\
 \Rightarrow WP(f_i(R),t_L;R_i) + t_L &\geq \max_{A \in \mathcal{X}} \sum_{j \ne i} WP(A_j,t_L; R_j)- \sum_{j \ne i} WP(f_j(R),t_L;R_j) + t_L =p_i(R).
 \end{align*}
 But this implies that
 \begin{align*}
 \Big(f_i(R),p_i(R)\Big)~R_i~\Big(f_i(R),WP(f_i(R),t_L;R_i)+t_L\Big)~I_i~(\emptyset,t_L),
\end{align*}
where the second relation comes from the definition of $WP$.

 Suppose $f_i(R)$ is not an acceptable bundle at $R_i$, then $(f_i(R),p_i(R))~I_i~(\emptyset,p_i(R))$.
 Then, the relation (\ref{eq:tl}) implies that $t_L \ge p_i(R)$. But by construction, $p_i(R) \ge t_L$. Hence,
 $p_i(R) = t_L$ if $f_i(R) \notin \mathcal{S}_i$.
\end{proof}

Using Claim \ref{cl:loser}, we prove each assertion of the proposition. \\

\noindent {\sc Proof of (1).} We prove that the GVCG-$t_L$ is DSIC. Fix agent $i \in N$, $R_{-i} \in \mathcal{T}^{n-1}$, and $R_i,R'_i \in \mathcal{T}$.
Let $A \equiv f(R_i,R_{-i})$ and $A' \equiv f(R'_i,R_{-i})$. We start with a simple lemma.

\begin{lemma}\label{lem:lvcg}
If $A_i$ and $A'_i$ belong to the acceptable bundle set at $R_i$, then
$$p_i(R_i,R_{-i}) \leq p_i(R'_i,R_{-i}).$$
\end{lemma}
\begin{proof}
Note that
\begin{align*}
p_i(R_i,R_{-i}) - p_i(R'_i,R_{-i}) &= \big[ \max_{\hat{A} \in \mathcal{X}} \sum_{j \ne i} WP(\hat{A}_j,t_L;R_j) - \sum_{j \ne i} WP(A_j,t_L;R_j)\big] \\
&- \big[ \max_{\hat{A} \in \mathcal{X}} \sum_{j \ne i} WP(\hat{A}_j,t_L;R_j)- \sum_{j \ne i} WP(A'_j,t_L;R_j)\big] \\
&= \sum_{j \ne i} WP(A'_j,t_L;R_j) - \sum_{j \ne i} WP(A_j,t_L;R_j) \\
&= WP(A'_i,t_L;R_i) + \sum_{j \ne i} WP(A'_j,t_L;R_j) \\
&- WP(A_i,t_L;R_i) - \sum_{j \ne i} WP(A_j,t_L;R_j) \\
&= \sum_{j \in N}WP(A'_j,t_L;R_j) - \sum_{j \in N}WP(A_j,t_L;R_j) \\
&\le 0,
\end{align*}
where the third equality follows from the fact that $A_i,A'_i$ belong to the acceptable bundle set at $R_i$ and the last inequality
follows from the fact that $f(R)=A$.
\end{proof}

Let $\mathcal{S}_i$ be the acceptable bundle set of agent $i$ according to $R_i$. We consider two cases. \\

\noindent{\sc Case 1.} $A_i \in \mathcal{S}_i$. If $A'_i \in \mathcal{S}_i$, then Lemma \ref{lem:lvcg} implies that
$$(A_i,p_i(R_i,R_{-i}))~I_i~(A'_i,p_i(R_i,R_{-i}))~R_i~(A'_i,p_i(R'_i,R_{-i})).$$
If $A'_i \notin \mathcal{S}_i$, then Equation (\ref{eq:tl2}) implies that $p_i(R'_i,R_{-i})=t_L$. But, then Inequality (\ref{eq:tl}) implies that
$$(A_i,p_i(R_i,R_{-i}))~R_i~(\emptyset,t_L)~I_i~(A'_i,t_L).$$

\noindent{\sc Case 2.} $A_i \notin \mathcal{S}_i$. By Equation \ref{eq:tl2}, $p_i(R_i,R_{-i})=t_L$.
Now, note that since $A_i \notin \mathcal{S}_i$, we have $WP(A_i,t_L;R_i)=0$, and hence,
$$\sum_{j \in N} WP(A_j,t_L;R_j)=\max_{\hat{A} \in \mathcal{X}} \sum_{j \not=i} WP(\hat{A}_j,t_L;R_j).$$
This implies that

$$\sum_{j \in N} WP(A'_j,t_L;R_j) \leq \sum_{j \in N} WP(A_j,t_L;R_j)=\max_{\hat{A} \in \mathcal{X}} \sum_{j \not=i} WP(\hat{A}_j,t_L;R_j),$$
where the first inequality followed from the definition of $A$. This implies that
$$WP(A'_i,t_L;R_i) \leq \max_{\hat{A} \in \mathcal{X}} \sum_{j \not=i} WP(\hat{A}_j,t_L;R_j)-\sum_{j \ne i} WP(A'_j,t_L;R_j) =p_i(R'_i,R_{-i})-t_L.$$
This further implies that $$\Big(A_i,p_i(R_i,R_{-i})\Big)~I_i~(\emptyset,t_L)~I_i~\Big(A'_i,WP(A'_i,t_L;R_i)+t_L\Big)~R_i~\Big(A'_i,p_i(R'_i,R_{-i})\Big).$$

Hence, in both cases, we see that agent $i$ prefers his outcome $(A_i,p_i(R_i,R_{-i}))$ in the GVCG mechanism to
the outcome obtained by reporting $R'_i$. This concludes the proof that the GVCG-$t_L$ is strategy-proof. \\

\noindent {\sc Proofs of (2) and (3).} By Inequality (\ref{eq:tl}), for every $i \in N$ and for every $R$, we have
$\Big(f_i(R),p_i(R)\Big)~R_i~(\emptyset,t_L)$. If $t_L \le 0$, we get that $\Big(f_i(R),p_i(R)\Big)~R_i~(\emptyset,0)$, which
is individual rationality. (3) follows from (2). \\

\noindent {\sc Proof of (4).} We now show that for $n=2$, the GVCG-$t_L$ mechanism (for any $t_L \in \mathbb{R}$)
is Pareto efficient in any dichotomous domain. Let $N=\{1,2\}$
and consider a preference profile $R \equiv (R_1,R_2)$ with $\mathcal{S}_1$ and $\mathcal{S}_2$ as
the collection of acceptable bundles of agents 1 and 2 respectively. We consider two cases. As before, denote
by $(f,\mathbf{p}) \equiv (f,\mathbf{p}^{vcg,t_L})$. \\

\noindent {\sc Case 1.} There exists $S_1 \in \mathcal{S}_1$
and $S_2 \in \mathcal{S}_2$ such that $S_1 \cap S_2 = \emptyset$. Then,
$f_1(R) \in \mathcal{S}_1$ and $f_2(R) \in \mathcal{S}_2$ and $p_1(R)=p_2(R)=t_L$.
Denote $A^*_1:=f_1(R)$ and $A^*_2:=f_2(R)$.
Assume for contradiction that there is an outcome profile $((A_1,p_1), (A_2,p_2))$ such that
$p_1+p_2 \ge 2t_L$, $(A_1,p_1)~R_1~(A^*_1,t_L)$, and $(A_2,p_2)~R_2~(A^*_2,t_L)$ with strict inequality
holding for one of them. By the last two relations, it must be that $p_1 \le t_L$ and $p_2 \le t_L$ with strict
inequality holding whenever these relations are strict, which
means that $p_1+p_2 \le 2t_L$. But this means $p_1+p_2 = 2t_L$ since we assumed $p_1+p_2 \ge 2t_L$.
Hence, none of the relations can hold strict, a contradiction. \\

\noindent {\sc Case 2.} For every $S_1 \in \mathcal{S}_1$ and for every $S_2 \in \mathcal{S}_2$,
we have $S_1 \cap S_2 \ne \emptyset$. Then, one of the agents in $\{1,2\}$ will be assigned
an acceptable bundle in $f$. Let this agent be $1$. Hence, $f_1(R) \in \mathcal{S}_1$
and $f_2(R) = \emptyset$. Further, $p_1(R) = w_2(t_L)+t_L$, where $w_2(t_L)$ is the willingness
to pay of agent $2$ at $t_L$, and $p_2(R)=t_L$.

Denote $A^*_1:=f_1(R)$ and assume for contradiction that there is an
outcome profile $((A_1,p_1), (A_2,p_2))$ such that
$p_1+p_2 \ge w_2(t_L)+2t_L$, $(A_1,p_1)~R_1~(A^*_1,w_2(t_L)+t_L)$, and $(A_2,p_2)~R_2~(\emptyset,t_L)$ with strict inequality
holding for one of them. Consider the following two subcases - by our assumption that
for every $S_1 \in \mathcal{S}_1$ and for every $S_2 \in \mathcal{S}_2$,
we have $S_1 \cap S_2 \ne \emptyset$, only the following two subcases may happen. \\

\begin{itemize}

\item {\sc Case 2a.} Suppose $A_1 \in \mathcal{S}_1$ and $A_2 \notin \mathcal{S}_2$.
Since $(A_1,p_1)~R_1~(A^*_1,w_2(t_L)+t_L)$ and $(A_2,p_2)~R_2~(\emptyset,t_L)$,
we have $p_1 \le w_2(t_L)+t_L$ and
$p_2 \le t_L$. Hence, we have $p_1+p_2 \le w_2(t_L)+2t_L$.

\item {\sc Case 2b.} Suppose $A_1 \notin \mathcal{S}_1$ and $A_2 \in \mathcal{S}_2$.
Inequality (\ref{eq:tl}) implies $(A_1,p_1)~R_1~(A^*_1,w_2(t_L)+t_L)~R_1~(\emptyset,t_L)$. Hence, $p_1 \le t_L$.
Similarly, Inequality (\ref{eq:tl}) for agent $2$ implies that $p_2 \le w_2(t_L)+t_L$. Hence, again we have
$p_1+p_2 \le w_2(t_L)+2t_L$.

\end{itemize}
Both the cases imply that $p_1+p_2 \le w_2(t_L)+2t_L$ with strict inequality holding if
$$(A_1,p_1)~P_1~\Big(A^*_1,w_2(t_L)+t_L\Big)~\textrm{or}~(A_2,p_2)~P_2~(\emptyset,t_L).$$
But we are given that $p_1+p_2 > w_2(t_L)+2t_L$ or
$(A_1,p_1)~P_1~(A^*_1,w_2(t_L))$ or $(A_2,p_2)~P_2~(\emptyset,t_L)$. This is a contradiction. \\

\noindent {\sc Proof of (5).} We show the impossibility for $N=\{1,2,3\}$ and $M=\{a,b\}$. The impossibility
can be extended easily to the case when $n > 3$ and $m > 2$ by (i) considering preference profiles where each agent
$i$ has minimal acceptable bundle set $\mathcal{S}_i^{min} \subseteq \{a,b\}$ and (ii) every agent $i \notin \{1,2,3\}$
has arbitrarily small willingness to pay (at every transfer level) on acceptable bundles. This is similar as in
the proof of Theorem \ref{theo:impos}.

Fix the GVCG-$t_L$ mechanism for some $t_L \in \mathbb{R}$ and denote it as $(f,\mathbf{p}) \equiv (f^{vcg,t_L},\mathbf{p}^{vcg,t_L})$.
Consider the following single-minded preference profile $(R_1,R_2,R_3)$ such that
$$\mathcal{S}^{min}_1=\{a\}, \mathcal{S}^{min}_2=\{b\}, \mathcal{S}^{min}_3=\{a,b\}.$$
The WP values at transfer level $t_L$ are as follows:
\begin{align*}
WP(\{a\},t_L;R_1) = w_1; WP(\{b\},t_L;R_2) = w_2; WP(\{a,b\},t_L;R_3) = w_3,
\end{align*}
such that $w_1+w_2 > w_3 > \max(w_1,w_2)$. Further, we require $R_1$ and $R_2$ to satisfy the following:
\begin{align*}
\big(\{a\},w_3 - w_2 + t_L\big)~I_1~\big(\emptyset,t_L - \epsilon\big)~\textrm{and}\big(\{b\},w_3-w_1+t_L\big)~I_2~\big(\emptyset,t_L - \epsilon\big).
\end{align*}
Such dichotomous preferences $R_1,R_2,R_3$ are possible to construct. Figure \ref{fig:impos_pe} illustrates the some indifference vectors of $R_1, R_2$, and $R_3$.
\begin{figure}[!tbh]
\centering
\includegraphics[height=3in]{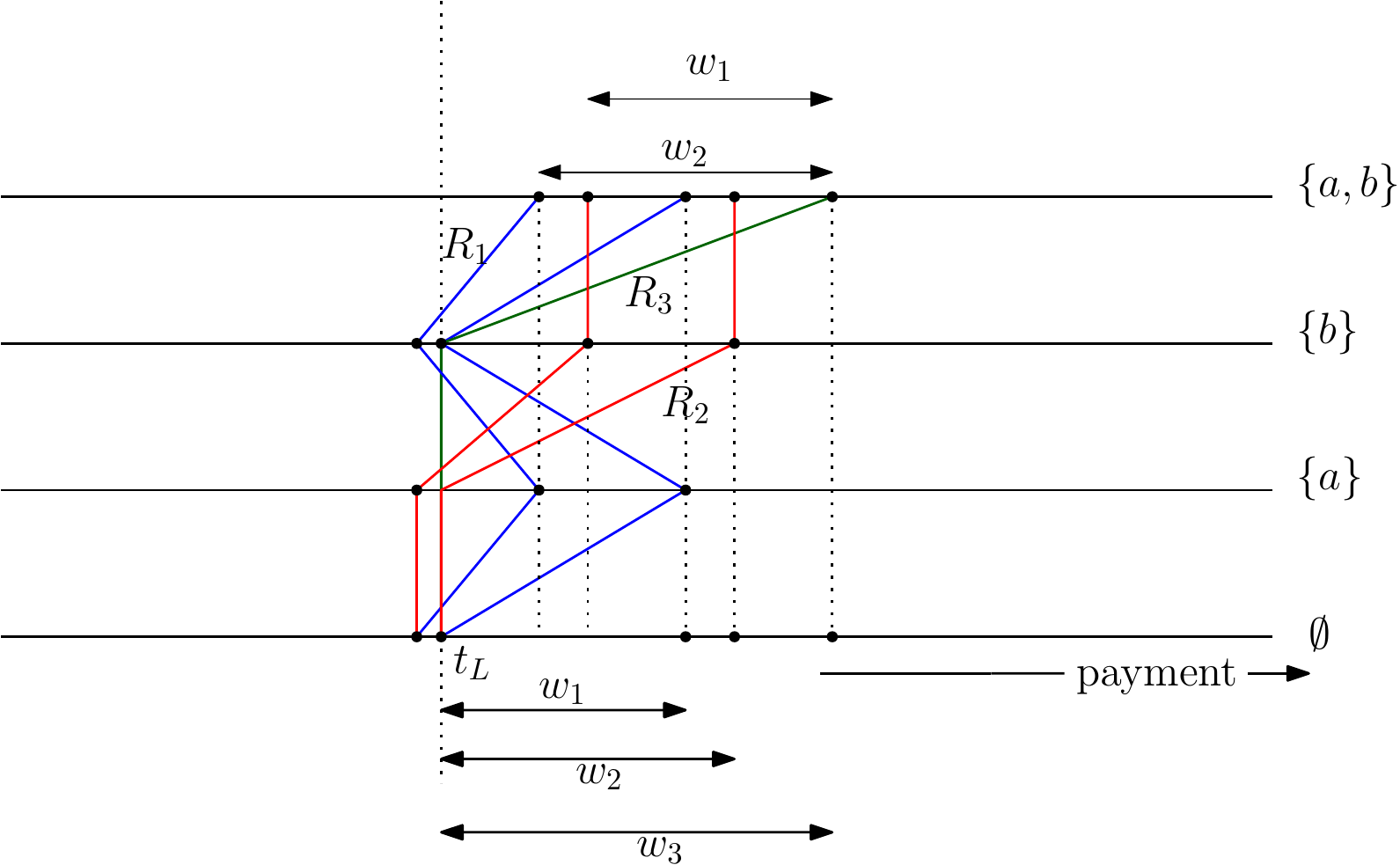}
\caption{A profile of dichotomous preferences for $N=\{1,2,3\}$ and $M=\{a,b\}$.}
\label{fig:impos_pe}
\end{figure}

Hence, the GVCG-$t_L$ mechanism produces the following outcome:
\begin{align*}
f_1(R_1,R_2,R_3)&=\{a\},~~f_2(R_1,R_2,R_3)=\{b\},~~f_3(R_1,R_2,R_3)=\emptyset; \\
p_1(R_1,R_2,R_3)&=w_3-w_2+t_L,~~p_2(R_1,R_2,R_3)=w_3 - w_1+t_L,~~p_3(R_1,R_2,R_3)=t_L.
\end{align*}
Consider the following outcome profile
\begin{align*}
z_1:=(\emptyset,t_L - \epsilon); z_2:=(\emptyset,t_L-\epsilon); z_3:= (\{a,b\},w_3+t_L).
\end{align*}
By construction (see Figure \ref{fig:impos_pe}), each agent $i \in \{1,2,3\}$ is indifferent
between $z_i$ and $(f_i(R),p_i(R))$.
Total transfers in the outcome profile $z$ is: $w_3+ 3t_L - 2\epsilon$.
Total transfers in the GVCG-$t_L$ mechanism: $2w_3 - (w_1+w_2) + 3t_L < w_3 + 3t_L - \epsilon$, where
the inequality follows since $w_3 < w_1+w_2$ and $\epsilon > 0$ is arbitrarily close to zero. Hence, the GVCG-$t_L$ mechanism
is not Pareto efficient.
\end{proof}

\subsection{Proof of Theorem \ref{theo:gvcg}}

\begin{proof}
By Proposition \ref{prop:vcg}, the GVCG mechanism is DSIC, individually rational, and satisfies no subsidy.
Now, we prove Pareto efficiency. Let $\mathcal{T}$ be a dichotomous domain
satisfying positive income effect. Assume for contradiction that there exists a profile $R \in \mathcal{T}^n$ such that $ (f^{vcg}(R),\mathbf{p}^{vcg}(R))$
is not Pareto efficient. As before, let $(\mathcal{S}_i,w_i)$ denote the dichotomous preference $R_i$ of any agent $i$.
Let  $f^{vcg}(R) \equiv A $ and $\mathbf{p}^{vcg}(R) \equiv (p_1,\ldots,p_n)$. Then there exists, an outcome profile $((A'_1,p'_1),\ldots,(A'_n,p'_n))$ which
Pareto dominates $((A_1,p_1),\ldots,(A_n,p_n))$.

We consider various cases to derive relationship between $p_i$ and $p'_i$ for each $i \in N$. \\

\noindent {\sc Case 1.} Pick $i \in N$ such that $A_i, A'_i \in \mathcal{S}_i$
or $A_i,A'_i \notin \mathcal{S}_i$. Dichotomous preference implies that $(A'_i,p'_i)~I_i~(A_i,p'_i)$.
But $(A'_i,p'_i)~R_i~(A_i,p_i)$ implies that $(A_i,p'_i)~R_i~(A_i,p_i)$. Hence, we get
\begin{align}\label{eq:p1}
p_i &\ge p'_i~\forall~i~\textrm{such that}~A_i,A'_i \in \mathcal{S}_i~\textrm{or}~A_i,A'_i \notin \mathcal{S}_i.
\end{align}

\noindent {\sc Case 2.} Pick $i \in N$ such that $A_i \notin \mathcal{S}_i$ but $A'_i \in \mathcal{S}_i$.
This implies that $p_i=0$ (by Lemma \ref{lem:ir}).
Hence, $(A'_i,p'_i)~R_i~(A_i,p_i)~I_i~(A_i,0)~I_i~(\emptyset,0)~I_i~(A'_i,w_i(0)).$
Thus,
\begin{align}\label{eq:p2}
w_i(0) + p_i &\ge p'_i~\forall~i~\textrm{such that}~A_i \notin \mathcal{S}_i, A'_i \in \mathcal{S}_i.
\end{align}

\noindent {\sc Case 3.} Pick $i \in N$ such that $A_i \in \mathcal{S}_i$ but $A'_i \notin \mathcal{S}_i$.
Since $A'_i \notin \mathcal{S}_i$, we can write $(A'_i,p'_i)~I_i~(\emptyset,p'_i)~I_i~(A_i,p'_i+w_i(p'_i))$.
But $(A'_i,p'_i)~R_i~(A_i,p_i)$ implies that
$$p_i \ge p'_i+w_i(p'_i).$$
Also, $(\emptyset,p'_i)~I_i~(A'_i,p'_i)~R_i~(A_i,p_i)~R_i~(\emptyset,0)$, where the last inequality
is due to individual rationality of the GVCG mechanism. Hence, $p'_i \le 0$. But then, positive income effect implies
that $w_i(p'_i) \ge w_i(0)$. This gives us
\begin{align}\label{eq:p3}
p_i &\ge p'_i + w_i(0)~\forall~i~\textrm{such that}~A_i \in \mathcal{S}_i, A'_i \notin \mathcal{S}_i.
\end{align}

By summing over Inequalities \ref{eq:p1}, \ref{eq:p2}, and \ref{eq:p3},  we get
\begin{align*}
\sum_{i \in N} p_i &\ge \sum_{i \in N} p'_i + \sum_{i: A_i \in \mathcal{S}_i, A'_i \notin \mathcal{S}_i}w_i(0) - \sum_{i: A_i \notin \mathcal{S}_i, A'_i \in \mathcal{S}_i} w_i(0). \\
&= \sum_{i \in N} p'_i + \sum_{i: A_i \in \mathcal{S}_i, A'_i \notin \mathcal{S}_i}w_i(0) + \sum_{i: A_i,A'_i \in \mathcal{S}_i}w_i(0) - \sum_{i: A_i,A'_i \in \mathcal{S}_i}w_i(0) - \sum_{i: A_i \notin \mathcal{S}_i, A'_i \in \mathcal{S}_i} w_i(0). \\
& = \sum_{i \in N} p'_i + \sum_{i \in N}WP(A_i,0;R_i) - \sum_{i \in N}WP(A'_i,0;R_i) \\
&\ge \sum_{i \in N} p'_i,
\end{align*}
where the inequality follows from the definition of the GVCG mechanism. Also, note that the inequality above is strict if any of the Inequalities \ref{eq:p1}, \ref{eq:p2}, and \ref{eq:p3}
is strict. This contradicts the fact that the outcome $((A'_1,p'_1),\ldots,(A'_n,p'_n))$ Pareto dominates $((A_1,p_1),\ldots,(A_n,p_n))$.
\end{proof}

\subsection{Proof of Theorem \ref{theo:unique}}

\begin{proof}
Let $(f,\mathbf{p})$ be a Pareto efficient, DSIC, IR mechanism satisfying no subsidy.
The proof proceeds in two steps. We assume without loss of generality that
at every preference profile $R$, if an agent $i \in N$ is assigned an acceptable
bundle $f_i(R)$, then $f_i(R)$ is a {\em minimal} acceptable bundle at $R_i$, i.e., there
does not exist another acceptable bundle $S_i \subsetneq f_i(R)$ at $R_i$.~\footnote{This is without loss of
generality for the following reason. For every Pareto efficient, DSIC, IR mechanism $(f,\mathbf{p})$ satisfying
no subsidy, we can construct another mechanism $(f',\mathbf{p}')$ such that: for all $R$ and for all $i \in N$,
$f'_i(R) \subseteq f_i(R)$ and $f'_i(R)$ is a minimal acceptable bundle at $R_i$ whenever $f_i(R)$ is an acceptable bundle
at $R_i$ and $f'_i(R)=f_i(R)$ otherwise. Further, $\mathbf{p}'=\mathbf{p}$. It is routine to verify that
$(f',\mathbf{p}')$ is DSIC, IR, Pareto efficient and satisfies no subsidy. Finally, by construction, if $(f',\mathbf{p}')$ is a generalized VCG mechanism, then
$(f,\mathbf{p})$ is also a generalized VCG mechanism.} We now proceed with the proof in two Steps. \\

\noindent {\sc Allocation is GVCG allocation.} In this step, we argue that $f$ must satisfy:
$$f(R) \in \arg \max_{A \in \mathcal{X}} \sum_{i \in N}WP(A_i,0;R_i)~\forall~R \in \mathcal{T}^n$$
Assume for contradiction that for some $R \in \mathcal{T}^n$, we have
$$\sum_{i \in N}WP(f_i(R),0;R_i) < \max_{A \in \mathcal{X}} \sum_{i \in N}WP(A_i,0;R_i).$$
Before proceeding with the rest of the proof, we fix a generalized VCG mechanism $(f^{vcg},p^{vcg})$ and introduce a notation. For every $R'$, denote by
$$N_{0+}(R') := \Big\{i \in N: \big[(f^{vcg}_i(R'),p^{vcg}_i(R'))~I'_i~(\emptyset,0)\big]~\textrm{and}~\big[(f_i(R'),p_i(R'))~P'_i~(\emptyset,0)\big]\Big\}.$$

We now construct a sequence of preference profiles, starting with preference profile $R$, as follows.
Let $R^0:= R$. Also, we will maintain a sequence of subsets of agents, which is initialized as $B^0:=\emptyset$.
We will denote the preference profile constructed in step $t$ of the sequence as $R^t$ and the
willingness to pay map at preference $R^t_i$ as $w^t_i$ for each $i \in N$.
\\

\begin{enumerate}
\item[S1.] If $N_{0+}(R^t) \setminus B^t = \emptyset$, then stop. Else, go to the next step.

\item[S2.] Choose $k^t \in N_{0+}(R^t) \setminus B^t$ and consider $R^{t+1}_{k^t}$ to be a quasilinear dichotomous
preference with valuation $w_{k^t}^{t+1}(0) \in (p_{k^t}(R^t),w_{k^t}^t(0))$ and a unique minimal acceptable bundle $f_{k^t}(R^t)$ - such a quasilinear
preference exists because $\mathcal{T} \supseteq \mathcal{D}^{QL}$.
Let $R^{t+1}_j=R^t_j$ for all $j \ne k^t$.

\item[S3.] Set $B^{t+1}:= B^t \cup \{k^t\}$ and $t:= t+1$. Repeat from Step S1.

\end{enumerate}

Because of finiteness of number of agents, this process will terminate finitely in some $T < \infty$ steps. We establish some claims about the preference
profiles generated in this procedure.

\begin{claim}\label{cl:scl1}
For every $t \in \{0,\ldots,T-1\}$, $f_{k^t}(R^{t+1}) = f_{k^t}(R^t)$ and $p_{k^t}(R^{t+1})=p_{k^t}(R^t)$.
\end{claim}
\begin{proof}
Fix $t$ and assume for contradiction $f_{k^t}(R^{t+1}) \ne f_{k^t}(R^t)$. Since $f_{k^t}(R^t)$ is the unique minimal acceptable
bundle at $R^{t+1}_{k^t}$ and $f$ only assigns a minimal acceptable bundle whenever it assigns
acceptable bundles, it must be that $f_{k^t}(R^{t+1})$ is not an acceptable bundle at $R^{t+1}_{k^t}$.
Then, by Lemma \ref{lem:ir}, we get
$p_{k^t}(R^{t+1})=0$. Since $w_{k^t}^{t+1}(0) > p_{k^t}(R^t)$ and $f_{k^t}(R^t)$ is an acceptable bundle at $R^{t+1}_{k^t}$, we get
$$(f_{k^t}(R^t),p_{k^t}(R^t))~P_{k^t}^{t+1}~(\emptyset,0)~I_{k^t}^{t+1}(f_{k^t}(R^{t+1}),p_{k^t}(R^{t+1})).$$
This contradicts DSIC. Finally, if $f_{k^t}(R^{t+1}) = f_{k^t}(R^t)$, we must have $p_{k^t}(R^{t+1}) = p_{k^t}(R^t)$ due
to DSIC since acceptable bundle at $R^{t+1}_{k^t}$ is $f_{k^t}(R^t)$ and $f_{k^t}(R^t)$ is also an acceptable bundle
at $R^t_{k^t}$.
\end{proof}

The next claim establishes a useful inequality.
\begin{claim}\label{cl:scl2}
For every $t \in \{0,\ldots,T\}$, the following holds:
$$w^t_{k^t}(0)+\max_{A \in \mathcal{X}, A_{k^t}= f_{k^t}(R^t)} \sum_{j \ne k^t} WP_j(A_j,0; R^t_j) \le \max_{A \in \mathcal{X}} \sum_{j \ne k^t}  WP_j(A_j,0; R^t_j).$$
\end{claim}
\begin{proof}
Pick some $t \in \{0,\ldots,T\}$ and suppose the above inequality does not hold. We complete the proof in two steps. \\

\noindent {\sc Step 1.} In this step, we argue that $f^{vcg}_{k^t}$ must be an acceptable bundle for agent $k^t$ at preference $R^t$. If this is not
true, then we must have
\begin{align*}
\sum_{j \in N} WP(f^{vcg}_j(R^t),0;R^t_j) &= \sum_{j \ne k^t} WP(f^{vcg}_j(R^t),0;R^t_j) \\
&\le \max_{A \in \mathcal{X}} \sum_{j \ne k^t}WP(A_j,0;R^t_j) \\
&< w^t_{k^t}(0)+\max_{A \in \mathcal{X}, A_{k^t}= f_{k^t}(R^t)} \sum_{j \ne k^t} WP(A_j,0; R^t_j) \\
&= WP(f_{k^t}(R^t), 0;R^t_{k^t}) + \max_{A \in \mathcal{X}, A_{k^t}= f_{k^t}(R^t)} \sum_{j \ne k^t} WP(A_j,0; R^t_j),
\end{align*}
where the last inequality follows from our assumption that the claimed inequality does not hold and the last equality
follows from the fact that $f_{k^t}(R^t)$ is an acceptable bundle of agent $k^t$ at $R^t_{k^t}$. But, then the resulting inequality contradicts
the definition of $f^{vcg}$. \\

\noindent {\sc Step 2.} We complete the proof in this step. Notice that the payment of agent $k^t$ in $(f^{vcg},p^{vcg})$ is
defined as follows.
\begin{align*}
p^{vcg}_{k^t}(R^t) &= \max_{A \in \mathcal{X}} \sum_{j \ne k^t}WP(A_j,0;R^t_j) - \sum_{j \ne k^t} WP(f^{vcg}_j(R^t),0;R^t_j) \\
&< w^t_{k^t}(0)+\max_{A \in \mathcal{X}: A_{k^t}= f_{k^t}(R^t)} \sum_{j \ne k^t} WP(A_j,0; R^t_j) - \sum_{j \ne k^t} WP(f^{vcg}_j(R^t),0;R^t_j) \\
&= w^t_{k^t}(0)+\max_{A \in \mathcal{X}: A_{k^t}= f_{k^t}(R^t)} \sum_{j \ne k^t} WP(A_j,0; R^t_j) \\
&- \sum_{j \in N} WP(f^{vcg}_j(R^t),0;R^t_j) + WP(f^{vcg}_{k^t}(R^t),0;R^t_{k^t}) \\
&= w^t_{k^t}(0)+\max_{A \in \mathcal{X}, A_{k^t}= f_{k^t}(R^t)} \sum_{j \in N} WP(A_j,0; R^t_j) - \sum_{j \in N} WP(f^{vcg}_j(R^t),0;R^t_j) \\
&\le  w^t_{k^t}(0),
\end{align*}
where the strict inequality followed from our assumption and the last equality follows from the fact both $f_{k^t}(R^t)$ and $f^{vcg}_{k^t}(R^t)$ are
acceptable bundles for agent $k^t$ at $R^t_{k^t}$ (Step 1). But, this implies that $$(f^{vcg}_{k^t}(R^t),p^{vcg}_{k^t}(R^t))~P^t_{k^t}~(f^{vcg}_{k^t}(R^t), w^t_{k^t}(0))~I^t_{k^t}~(\emptyset,0).$$
This is a contradiction to the fact that $k^t \in N_{0+}(R^t)$. This completes the proof.
\end{proof}

We now establish an important claim regarding an inequality satisfied by the sequence of preferences generated.
\begin{claim}\label{cl:scl3}
For every $t \in \{0,\ldots,T\}$,
$$\sum_{j \in N}WP(f_j(R^t),0;R^t_j) < \sum_{j \in N}WP(f^{vcg}_j(R^t),0;R^t_j).$$
\end{claim}
\begin{proof}
The inequality holds for $t=0$ by assumption. We now use induction. Suppose the inequality holds for $t \in \{0,\ldots,\tau-1\}$.
We show that it holds for $\tau$. To see this, denote $k \equiv k^{\tau-1}$. By Claim \ref{cl:scl1}, we know
that $f_k(R^{\tau-1})=f_k(R^{\tau})$. Further, by definition, $f_k(R^{\tau})$ belongs to the acceptable bundle of $k$ at $R^{\tau}_k$ and $R^{\tau-1}_k$.
Now, observe the following:
\begin{align*}
\sum_{j \in N}WP(f_j(R^{\tau}),0;R^{\tau}_j) &= w^{\tau}_{k}(0) + \sum_{j \ne k}WP(f_j(R^{\tau}),0;R^{\tau}_j)~\qquad~\textrm{(follows from definition of $k$)} \\
&\le w^{\tau}_{k}(0) + \max_{A \in \mathcal{X}:A_k=f_k(R^{\tau-1})=f_k(R^{\tau})}\sum_{j \ne k}WP(A_j,0;R^{\tau}_j) \\
&= w^{\tau}_{k}(0) + \max_{A \in \mathcal{X}:A_k=f_k(R^{\tau-1})=f_k(R^{\tau})}\sum_{j \ne k}WP(A_j,0;R^{\tau-1}_j) \\
& \textrm{(using the fact that $R^{\tau}_j=R^{\tau-1}_j$ for all $j \ne k$)} \\
&\le w^{\tau}_{k}(0) - w^{\tau-1}_{k}(0) + \max_{A \in \mathcal{X}} \sum_{j \ne k}WP(A_j,0;R^{\tau-1}_j) \qquad~\textrm{(using Claim \ref{cl:scl2})} \\
&< \max_{A \in \mathcal{X}} \sum_{j \ne k}WP(A_j,0;R^{\tau-1}_j)~\qquad~\textrm{(using the fact that $w^{\tau}_{k}(0) < w^{\tau-1}_{k}(0)$)} \\
&= \max_{A \in \mathcal{X}} \sum_{j \ne k}WP(A_j,0;R^{\tau}_j) \\
&\le \max_{A \in \mathcal{X}} \sum_{j \in N}WP(A_j,0;R^{\tau}_j).
\end{align*}
\end{proof}

We now complete our claim that the allocation is the same as in a GVCG mechanism. Let $R^T \equiv R'$. Let $f^{vcg}(R')=A^{vcg}$ and $f(R')=A'$.
Partition the set of agents as follows.
\begin{align*}
N_{++} &:= \{i : WP_i(A^{vcg}_i,0;R'_i) = WP(A'_i,0;R'_i) > 0\} \\
N_{+-} &:= \{ i: WP_i(A^{vcg}_i,0;R'_i) > 0, WP(A'_i,0;R'_i) = 0 \} \\
N_{-+} &:= \{ i: WP_i(A^{vcg}_i,0;R'_i) = 0, WP(A'_i,0;R'_i) > 0\} \\
N_{--} &:= \{ i: WP_i(A^{vcg}_i,0;R'_i) = WP(A'_i,0;R'_i) = 0\}.
\end{align*}

Now, consider the following consumption bundle $Z$:
\begin{equation*}
Z_i := \left\{
\begin{array}{ll}
(A^{vcg}_i,p_i(R')) & \textrm{if}~i \in N_{++} \cup N_{--}\\
(A^{vcg}_i, p_i(R') - WP(A'_i,0;R'_i)) & \textrm{if}~i \in N_{-+} \\
(A^{vcg}_i,WP(A^{vcg}_i,0;R'_i)) & \textrm{if}~i \in N_{+-}
\end{array}
\right.
\end{equation*}
Notice that for each $i \in N_{++} \cup N_{--}$, we have $Z_i = (A^{vcg}_i,p_i(R'))~I'_i~(A'_i,p_i(R'))$.
For each $i \in N_{+-}$, we know that $WP(A'_i,0;R'_i)=0$ - this implies that $A'_i$
is not an acceptable bundle at $R'_i$. Hence, for all $i \in N_{+-}$, we have
$Z_i = (A^{vcg}_i,WP(A^{vcg}_i,0;R'_i))~I'_i~(\emptyset,0)~I'_i~(A'_i,p_i(R'))$, where
the last relation follows from Lemma \ref{lem:ir}. Finally, for all $i \in N_{-+}$, $WP_i(A^{vcg}_i,0;R'_i) = 0$ implies that $(A^{vcg}_i,p^{vcg}_i(R'))~I'_i~(\emptyset,0)$.
Then, for every $i \in N_{-+}$, either we have $(A'_i,p_i(R'))~I'_i~(\emptyset,0)$ or we have $i \in B^T$ (i.e., $R'_i$ is a quasilinear
preference). In the first case, $p_i(R') = WP(A'_i,0;R'_i) $ implies
$$(A^{vcg}_i,p_i(R') - WP(A'_i,0;R'_i))~I'_i~(A^{vcg}_i,0)~I'_i~(\emptyset,0)~I'_i~(A'_i,p_i(R')).$$
In the second case, quasilinearity of $R'_i$ implies $(A^{vcg}_i,p_i(R') - WP(A'_i,0;R'_i))~I'_i~(A'_i,p_i(R'))$.
This completes the argument that $Z_i~R'_i~(A'_i,p_i(R'))$ for every $i \in N$.

Now, observe the sum of payments across all agents in $Z$ is:
\begin{align*}
&\sum_{i \notin N_{+-}}p_i(R') - \sum_{i \in N_{-+}}WP(A'_i,0;R'_i) + \sum_{i \in N_{+-}}WP(A^{vcg}_i,0;R'_i) \\
& = \sum_{i \in N}p_i(R') - \sum_{i \in N_{-+}}WP(A'_i,0;R'_i) + \sum_{i \in N_{+-}}WP(A^{vcg}_i,0;R'_i) \\
& \textrm{(since $A'_i$ is not acceptable, Lemma \ref{lem:ir} implies $p_i(R')=0$ for all $i \in N_{+-}$)} \\
&= \sum_{i \in N}p_i(R') + \sum_{i \in N} WP(A^{vcg}_i,0;R'_i) - \sum_{i \in N} WP(A'_i,0;R'_i) \\
&> \sum_{i \in N}p_i(R'),
\end{align*}
where the last inequality follows from Claim \ref{cl:scl3}.

Hence, $Z$ Pareto dominates the outcome $(f(R'),p(R'))$, contradicting Pareto efficiency. We now proceed to the next step
to show that the payment in $(f,\mathbf{p})$ must also coincide with the generalized VCG outcome. \\

\noindent {\sc Payment is GVCG payment.} Fix a preference profile $R$. We now know that
$$f(R) \in \arg \max_{A \in \mathcal{X}} \sum_{i \in N}WP(A_i,0;R_i).$$
By Lemma \ref{lem:ir}, for every $i \in N$, if $f_i(R)=f^{vcg}_i(R)$ is not acceptable
for agent $i$, then $p_i(R)=p^{vcg}_i(R)=0$ - here, we assume, without loss of generality,
that $f(R')=f^{vcg}(R')$ for all $R'$.\footnote{Depending on how we break ties
for choosing a maximum in the maximization of sum of willingness to pay, we have a different generalized
VCG mechanism. This assumption ensures that we pick the generalized VCG mechanism
that breaks the ties the same way as $f$.} We now consider two cases. \\

\noindent {\sc Case 1.} Assume for contradiction that there exists $i \in N$ such that $f_i(R)$ is an acceptable bundle of agent $i$ and
\begin{align}\label{eq:con1}
p_i(R) > \max_{A \in \mathcal{X}} \sum_{j \ne i}WP(A_j,0;R_j) - \sum_{j \ne i}WP(f_j(R),0;R_j).
\end{align}
Now consider $R'_i$ with the set of acceptable bundles the same in $R_i$ and $R'_i$ but $WP(f_i(R),0;R'_i) < p_i(R)$ but arbitrarily
close to $p_i(R)$. Let $A' \equiv f(R'_i,R_{-i})$. We argue that $A'_i$ is an acceptable bundle (at $R'_i$). If not, then
$$\max_{A \in \mathcal{X}} \sum_{j \ne i}WP(A_j,0;R_j) \ge \sum_{j \ne i}WP(A'_j,0;R_j) = WP(A'_i,0;R'_i) + \sum_{j \ne i}WP(A'_j,0,R_j),$$
where we used the fact that $A'_i$ is not an acceptable bundle for $i$. But then, by construction of $R'_i$ and Inequality (\ref{eq:con1}),
we get
$$WP(f_i(R),0;R'_i) + \sum_{j \ne i}WP(f_j(R),0;R_j) > \max_{A \in \mathcal{X}} \sum_{j \ne i}WP(A_j,0;R_j) \ge WP(A'_i,0;R'_i) + \sum_{j \ne i}WP(A'_j,0,R_j),$$
which is a contradiction to our earlier step that $f$ is the same allocation as in the GVCG mechanism.
Hence, $A'_i$ is an acceptable bundle at $R'_i$. But, then $p_i(R)=p_i(R'_i,R_{-i})$ by DSIC (since
$f_i(R)$ is also an acceptable bundle at $R_i$ and the set of acceptable bundles at $R_i$ and $R'_i$ are the same).
Since $WP(A'_i,0;R'_i) < p_i(R)=p_i(R'_i,R_{-i})$, we get a contradiction to individual rationality. \\

\noindent {\sc Case 2.} Assume for contradiction that there exists $i \in N$ such that $f_i(R)$ is an acceptable bundle of agent $i$ and
\begin{align*}
p_i(R)  < p^{vcg}_i(R) = \max_{A \in \mathcal{X}} \sum_{j \ne i}WP(A_j,0;R_j) - \sum_{j \ne i}WP(f_j(R),0;R_j).
\end{align*}
Pick $R'_i$ such that the set of acceptable bundles at $R'_i$ and $R_i$ are the same but $WP(f_i(R),0;R'_i) \in (p_i(R),p^{vcg}_i(R))$.
Notice that if $f_i(R'_i,R_{-i})$ is not an acceptable bundle at $R'_i$, then his payment is zero (Lemma \ref{lem:ir}). In that case,
$WP(f_i(R),0;R'_i) > p_i(R)$ implies that $$(f_i(R),p_i(R))~P'_i~(\emptyset,0)~I'_i~(f_i(R'_i,R_{-i}),p_i(R'_i,R_{-i})),$$ contradicting DSIC. Hence, $f_i(R'_i,R_{-i})=f^{vcg}_i(R'_i,R_{-i})$
is an acceptable bundle at $R'_i$. This implies that $f^{vcg}_i(R'_i,R_{-i})$ is an acceptable bundle at $R'_i$. Since the generalized VCG
is DSIC, we get that $p^{vcg}_i(R)=p^{vcg}_i(R'_i,R_{-i})$. But $WP(f^{vcg}_i(R'_i,R_{-i}),0;R'_i) < p^{vcg}_i(R)=p^{vcg}_i(R'_i,R_{-i})$ is
a contradiction to IR of the generalized VCG. This completes the proof.
\end{proof}

\subsection{Proof of Theorem \ref{theo:robustu}}

\begin{proof}
Assume for contradiction that $(f,\mathbf{p})$ is a desirable mechanism on $\mathcal{T}^n$.
By heterogeneous demand, there exist objects $a$ and $b$ such that
$0 < WP(a,0;R_0) < WP(b,0;R_0)$.
Consider a preference profile $R \in \mathcal{T}^n$ as follows:

\begin{enumerate}

\item Agent $1$ has quasilinear dichotomous preference with $\mathcal{S}_i^{min}=\{\{a,b\}\}$ and
value $w_1(0)$ that satisfies
\begin{align}\label{eq:wep3}
WP(\{a,b\},0;R_0) < w_1(0) < WP(\{a\},0;R_0)+WP(\{b\},0;R_0).
\end{align}

\item $R_i=R_0$ for all $i \in \{2,3\}$.

\item If $m > 2$, agent $4$ has quasilinear dichotomous preference with acceptable bundle $M \setminus \{a,b\}$
and value {\em very high}. If $m = 2$, agent $4$ has quasilinear dichotomous preference with acceptable
bundle $M$ and value equals to $\epsilon$, which is very close to zero.

\item For all $i > 4$, let $R_i$ be a quasilinear dichotomous preference with $\mathcal{S}^{min}_i=\{M\}$
and value equals to $\epsilon$, which is very close to zero.

\end{enumerate}

We begin by a useful claim.
\begin{claim}
\label{cl:use1}
Pick $k \in \{2,3\}$ and $x \in \{a,b\}$.
Let $R'$ be a preference profile such that
$R'_i = R_i$ for all $i \ne k$. Suppose $R'_k$ is such that
\begin{align}
\label{eq:as1}
WP(\{x\},0;R'_k) + WP(\{a,b\} \setminus \{x\},0;R_0) > w_1(0) > WP(\{a,b\},0;R'_k).
\end{align}
Then, the following are true:
\begin{enumerate}
\item $f_1(R')=\emptyset$

\item $f_2(R') \cup f_3(R') = \{a,b\}$

\item $f_2(R') \ne \emptyset$ and $f_3(R') \ne \emptyset$.

\end{enumerate}
\end{claim}
\begin{proof}
It is without loss of generality (due to Pareto efficiency) that $f_i(R') = \emptyset$ or $f_i(R') \in \mathcal{S}^{min}_i$ for all
$i$ who has dichotomous preference. Since $\epsilon$ is very close to zero, Pareto efficiency implies that (a) if $m=2$, $f_i(R')=\emptyset$ for all $i > 3$;
and (b) if $m > 2$, since agent $4$ has very high value for $M \setminus \{a,b\}$, $f_4(R')=M \setminus \{a,b\}$ and $f_i(R')=\emptyset$
for all $i > 4$. Hence, agents $1,2,$ and $3$ will be
allocated $\{a,b\}$ at $R'$. Denote $y \equiv \{a,b\} \setminus \{x\}$ and $\ell \equiv \{2,3\} \setminus \{k\}$. \\

\noindent {\sc Proof of (1) and (2).} Assume for contradiction $f_1(R') \ne \emptyset$. Pareto efficiency implies that $f_1(R')=\{a,b\}$
and $f_2(R')=f_3(R')=\emptyset$. Lemma \ref{lem:ir} implies that $p_2(R')=p_3(R')=0$. Then, consider
the following outcome:
$$z_1:= \Big(\emptyset,p_1(R')-w_1(0)\Big),~z_k:=\Big(\{x\},WP(\{x\},0;R'_k)\Big),~z_{\ell}:=\Big(\{y\},WP(\{y\},0;R'_{\ell})\Big),$$
$$z_i:=\Big(f_i(R'),p_i(R')\Big)~\forall~i > 3.$$
By definition of willingness to pay, $z_i~I_i~(\emptyset,0) \equiv \Big(f_i(R'),p_i(R')\Big)$ for all $i \in \{2,3\}$.
Since agent $1$ has quasilinear preferences, she is also indifferent between $z_1$ and $\Big(\{a,b\}, p_1(R')\Big) \equiv \Big(f_1(R'),p_1(R')\Big)$.
Thus, the difference in total
payment between the outcome $z$ and the payment in $(f,\mathbf{p})$ at $R'$ is
\begin{align*}
WP(\{x\},0;R'_k)+WP(\{y\},0;R'_{\ell}) - w_1(0) = WP(\{x\},0;R'_k)+WP(\{y\},0;R_0) - w_1(0) > 0,
\end{align*}
where the inequality follows from Inequality (\ref{eq:as1}). This is a contradiction to Pareto efficiency of $(f,\mathbf{p})$. Hence, $f_1(R)=\emptyset$.
By Pareto efficiency, $f_2(R') \cup f_3(R') = \{a,b\}$. \\

\noindent {\sc Proof of (3).}
Now, we show that $f_2(R') \ne \emptyset$ and $f_3(R') \ne \emptyset$.
Suppose $f_3(R')=\emptyset$. Then, $f_2(R')=\{a,b\}$ and Lemma \ref{lem:ir} implies that $p_3(R')=0$.
We first argue that $p_2(R')=WP(\{a,b\},0;R'_2)$. To see this, consider a quasilinear dichotomous preference $\tilde{R}_2$
with acceptable bundle $\{a,b\}$ and value equal to $WP(\{a,b\},0;R'_2)$. Notice that $w_1(0) > WP(\{a,b\},0;R'_2)$ - if $k=2$,
then this is true by Inequality (\ref{eq:as1}) and if $\ell=2$, then $R'_{\ell}=R_0$ satisfies $w_1(0) > WP(\{a,b\},0;R_0)$ by Inequality (\ref{eq:wep3}).
Since agents $1$ and $2$ have the same acceptable
bundle at $(\tilde{R}_2,R'_{-2})$ but $w_1(0) > WP(\{a,b\},0;R'_2)$, this implies that (due to Pareto efficiency), $f_2(\tilde{R}_2,R'_{-2})=\emptyset$
and $p_2(\tilde{R}_2,R'_{-2})=0$ (Lemma \ref{lem:ir}).
By DSIC, $(\emptyset,0)~\tilde{R}_2~(\{a,b\},p_2(R')).$ This implies that $WP(\{a,b\},0;R'_2) \le p_2(R')$. IR of agent $2$
at $R'$ implies $WP(\{a,b\},0;R'_2) = p_2(R')$.

Next, consider the following
outcome
$$z'_k:=(\{x\},WP(\{x\},0;R'_k), z'_{\ell}:=(\{y\},WP(\{y\},0;R'_{\ell}), z'_i:=(f_i(R'),p_i(R'))~\forall~i \notin \{2,3\}.$$
By definition, for every agent $i$, $z'_i~I'_i~(f_i(R'),p_i(R'))$.
The difference between the sum of payments of agents in $z'$ and $(f,\mathbf{p})$ at $R$ is:
\begin{align*}
WP(\{x\},0;R'_k)+WP(\{y\},0;R'_{\ell}) - p_2(R') &= WP(\{x\},0;R'_k)+WP(\{y\},0;R_0) - WP(\{a,b\},0;R'_2) \\
&> w_1(0) - WP(\{a,b\},0;R'_2) \\
&> 0,
\end{align*}
where the first inequality follows from Inequality (\ref{eq:as1}) and the second inequality follows from Inequality (\ref{eq:as1})
if $k=2$ and from Inequality (\ref{eq:wep3}) if $\ell=2$.
This contradicts Pareto efficiency of $(f,\mathbf{p})$. A similar proof shows that $f_2(R') \ne \emptyset$.
\end{proof}

Now, pick any $k \in \{2,3\}$ and set $R'_k=R_0$ in Claim \ref{cl:use1}. By Inequality (\ref{eq:wep3}), Inequality (\ref{eq:as1})
holds for $R_0$. As a result, we get that $f_2(R) \ne \emptyset, f_3(R) \ne \emptyset$, and $f_2(R) \cup f_3(R)= \{a,b\}$.
Hence, without loss of generality, assume that $f_2(R)=\{a\}$ and $f_3(R)=\{b\}$.\footnote{Since we have assumed $WP(\{b\},0;R_0) > WP(\{a\},0;R_0)$, this
may appear to be with loss of generality. However, if we have $f_2(R)=\{b\}$ and $f_3(R)=\{a\}$, then
we will swap $2$ and $3$ in the entire argument following this.}
We now complete the proof in two steps. \\

\noindent {\sc Step 1.} We argue that $p_2(R)=w_1(0) - WP(\{b\},0;R_0)$ and $p_3(R)=w_1(0) - WP(\{a\},0;R_0)$.
Suppose $p_2(R) > w_1(0) - WP(\{b\},0;R_0)$. Then, consider the quasilinear dichotomous preference $R^Q_2$ such that the minimum acceptable bundle
of agent $2$ is $\{a\}$ and his value $v$ satisfies
\begin{align}\label{eq:wpe2}
w_1(0) - WP(\{b\},0;R_0) < v < p_2(R).
\end{align}

Now, note that by IR of agent $2$ at $R$, we have $$p_2(R) \le WP(\{a\},0;R_0) \le WP(\{a,b\},0;R_0) < w_1(0),$$
where the strict inequality followed from Inequality (\ref{eq:wep3}).
Hence, $v < w_1(0)$ and $w_1(0) < v+WP(\{b\},0;R_0)$ by Inequality (\ref{eq:wpe2}). Hence, choosing
$k=2,$ $ x=a$ and $R'_k=R_2^Q$, we can apply Claim \ref{cl:use1} to conclude that $f_2(R^Q_2,R_{-2}) \cup f_3(R^Q_2,R_{-2}) = \{a,b\}$
and $f_2(R^Q_2,R_{-2}) \ne \emptyset$, $f_3(R^Q_2,R_{-2}) \ne \emptyset$. Since $R^Q_2$ is a dichotomous preference
with acceptable bundle $\{a\}$, Pareto efficiency implies that $f_2(R^Q_2)=\{a\}=f_2(R)$.
By DSIC, $p_2(R)=p_2(R^Q_2,R_{-2})$. But Inequality (\ref{eq:wpe2}) gives $v < p_2(R)=p_2(R_2^Q,R_{-2})$, and this contradicts individual rationality.

Next, suppose $p_2(R) < w_1(0) - WP(\{b\},0;R_0)$. Then, consider the quasilinear dichotomous preference $\hat{R}^Q_2$ such that
the minimal acceptable bundle of agent $2$ is $\{a\}$ and his value $\hat{v}$ satisfies
\begin{align}\label{eq:dd1}
p_2(R) < \hat{v} < w_1(0)-WP(\{b\},0;R_0).
\end{align}
Now, consider the preference profile $\hat{R}$ such that $\hat{R}_2=\hat{R}_2^Q$ and $\hat{R}_i=R_i$ for all $i \ne 2$.
We first argue that $f_2(\hat{R})=\emptyset$. Suppose not, then by Pareto efficiency, $f_2(\hat{R})=\{a\}$.
By Pareto efficiency, we have $f_3(\hat{R})=\{b\}$ and $f_1(\hat{R})= \emptyset$.
By Lemma \ref{lem:ir}, $p_1(\hat{R})=0$. We argue that $p_3(\hat{R})=WP(\{b\},0;R_0)$.
To see this, consider a profile $\hat{R}'$ where $\hat{R}'_i=\hat{R}_i$ for all $i \ne 3$
and $\hat{R}'_3$ is a quasilinear dichotomous preferences with minimum acceptable bundle $\{b\}$ and value
equal to $WP(\{b\},0;R_0)$ - notice that every agent in $\hat{R}'$ has quasilinear preference. As a result, Theorem \ref{theo:unique}
implies that the outcome of $(f,\mathbf{p})$ at $\hat{R}'$ must coincide with the GVCG mechanism. But $w_1(0) > \hat{v} +WP(\{b\},0;R_0)$
implies that $f_1(\hat{R}')=\{a,b\}$ and $f_2(\hat{R}')=f_3(\hat{R'})=\emptyset$. Then, DSIC implies that (incentive
constraint of agent $3$ from $\hat{R}'$ to $\hat{R}$) ~$0 \ge WP(\{b\},0;R_0) - p_3(\hat{R})$. By individual
rationality of agent $3$ at $\hat{R}$ we get, $p_3(\hat{R}) \le WP(\{b\},0;R_0)$, and combining these we get $p_3(\hat{R})=WP(\{b\},0;R_0)$.

Now, consider the following allocation vector $\hat{z}$:
$$\hat{z}_1:= \Big(\{a,b\},w_1(0)\Big), \hat{z}_2:= \Big(\emptyset, p_2(\hat{R}) - \hat{v}\Big), \hat{z}_3 := \Big(\emptyset, 0\Big),$$
$$\hat{z}_i := \Big(f_i(\hat{R}),p_i(\hat{R})\Big)~\forall~i > 3.$$
By definition of $w_1(0)$, we get that $\hat{z}_1~\hat{I}_1~(\emptyset,0)$. Also, since $\hat{R}_2$ is quasilinear with value $\hat{v}$, we get
$(\emptyset, p_2(\hat{R})-\hat{v})~\hat{I}_2~(\{a\},p_2(\hat{R})).$ For agent $3$, notice that $R_3=R_0$ and by the definition of willingness to pay,
we get
$(\emptyset,0)~\hat{I}_3~\Big(\{b\},WP(\{b\},0;R_0)\Big)$.
For $i > 3$, each agent $i$ gets the same outcome in $\hat{z}$ and $(f,\mathbf{p})$. Finally, the sum of payments
of agents $1, 2$, and $3$ (payments of other agents remain unchanged) in $\hat{z}$ is
$$w_1(0) + p_2(\hat{R}) - \hat{v} > p_2(\hat{R}) + p_3(\hat{R}),$$
where the strict inequality follows from Inequality (\ref{eq:dd1}) and the fact that $p_3(\hat{R})=WP(\{b\},0;R_0)$. This contradicts the fact that
$(f,\mathbf{p})$ is Pareto efficient.

Hence, we must have $f_2(\hat{R})=\emptyset$. By Lemma \ref{lem:ir}, we have $p_2(\hat{R})=0$. But since $v > p_2(R)$, we get
$(\{a\},p_2(R))~\hat{P}_2~(\emptyset,0)$. Hence, $(f_2(R),p_2(R))~\hat{P}_2~(f_2(\hat{R}),p_2(\hat{R}))$. This contradicts DSIC.

An identical argument establishes that $p_3(R)=w_1(0) - WP(\{a\},0;R_0)$. \\

\noindent {\sc Step 2.} In this step, we show that agent $2$ can manipulate at $R$, thus contradicting DSIC and completing the proof.
Consider a quasilinear dichotomous preference $\bar{R}^Q_2$ where the minimum acceptable bundle of agent $2$ is $\{b\}$ (note that
$f_2(R)=\{a\}$) and his value $\bar{v}$ is $WP(\{b\},0;R_0)$. Consider the preference profile $\bar{R}$ where $\bar{R}_2=\bar{R}^Q_2$ and $\bar{R}_i=R_i$ for all $i \ne 2$.
Notice that if we let $k=2,$  $x = b$, and $R'_k=\bar{R}^Q_2$, Inequality (\ref{eq:as1}) holds, and hence, Claim \ref{cl:use1} implies
that $f_2(\bar{R}) \ne \emptyset$ and $f_3(\bar{R}) \ne \emptyset$ but $f_2(\bar{R}) \cup f_3(\bar{R})= \{a,b\}$.
Hence, Pareto efficiency implies that $f_2(\bar{R})=\{b\}$ and $f_3(\bar{R})=\{a\}$.
Then, we can mimic the argument in Step 1 to conclude that $$p_2(\bar{R})=w_1(0) - WP(\{a\},0;R_0).$$

Now, by the definition of willingness to pay,
$$\Big(\{b\},WP(\{b\},0;R_0)\Big)~I_0~\Big(\{a\},WP(\{a\},0;R_0)\Big)$$ and by our assumption,
$WP(\{b\},0;R_0) > WP(\{a\},0;R_0)$. By subtracting $WP(\{a\},0;R_0)+WP(\{b\},0;R_0) - w_1(0)$ (which
is positive by Inequality (\ref{eq:wep3})) from payments on both sides, and using the fact that $R_0$ satisfies strict positive income effect, we get
$$\Big(\{b\},w_1(0) - WP(\{a\},0;R_0)\Big)~P_0~\Big(\{a\},w_1(0) - WP(\{b\},0;R_0)\Big).$$
Hence, $(f_2(\bar{R}),p_2(\bar{R}))~P_2~(f_2(R),p_2(R))$.
This contradicts DSIC.
\end{proof}


\end{document}